\documentclass[journal, twocolumn]{IEEEtran}

\usepackage{amssymb} 
\usepackage{amsmath} 
\usepackage{color}

\usepackage{theorem} 

\usepackage{latexsym} 
\usepackage{epsfig} 
\usepackage{psfrag} 
\usepackage{bm, bbm} 
\usepackage{xspace} 
\usepackage{graphicx} 
\usepackage{cite} 
\usepackage{url} 
\usepackage{subfigure,afterpage, wrapfig, pifont}

\renewcommand{\mathbf}[1]{{\bm{#1}}}

\def\Bbb{\mathbb} 
\def\Z{{\Bbb Z}} 
\def\Zp{{\Bbb Z}_{\geq 0}}

\newcommand{\tvomega}{\tilde{\boldsymbol{\omega}}}

\newcommand{\code}[1]{\mathsf{#1}} 
\newcommand{\matr}[1]{\mathbf{#1}} 
\newcommand{\vect}[1]{\mathbf{#1}} 
\newcommand{\set}[1]{\mathcal{#1}} 
\newcommand{\GF}[1]{\mathbb{F}_{#1}} 
\newcommand{\defeq}{\triangleq}

\newcommand{\tr}{\mathsf{T}}

  {\noindent \emph{Proof:}}{\hfill$\square$}

 \newtheorem{theorem}{Theorem}
 \newtheorem{lemma}[theorem]{Lemma} 

  \newtheorem{definition}[theorem]{Definition} 
  \newtheorem{example}[theorem]{Example} 
  \newtheorem{remark}[theorem]{Remark} 

\def\squarebox#1{\hbox to #1{\hfill\vbox to #1{\vfill}}}
\newcommand{\qed}{\hspace*{\fill}%
  \vbox{\hrule\hbox{\vrule\squarebox{.667em}\vrule}\hrule}\smallskip}

\newcommand{\defend}{\qed} 
\newcommand{\exend}{\qed} 

\newcommand{\remend}{\qed}

\newcommand{\vomega}{\boldsymbol{\omega}}

\newcommand{\codeCconv}{\code{C}_{\mathrm{conv}}} 
 
\newcommand{\codeCQC}[1]{\code{C}_{\mathrm{QC}}^{(r)}}

\newcommand{\matrHconv}{\matr{H}_{\mathrm{conv}}}

\newcommand{\matrHQC}[1]{\matr{H}_{\mathrm{QC}}^{(#1)}}

 \newcommand{\Ftwo}{\mathbb{F}_2} 
\newcommand{\Ftwoxmodr}{\Ftwo[X] / \langle X^r{-}1 \rangle}
\newcommand{\shortFtwoxmodr}{\Ftwo^{\langle r \rangle}[X]}

\newcommand{\codeCbar}{\overline{\code{C}}}
\newcommand{\codeCbarconv}{\overline{\code{C}}_{\mathrm{conv}}}
\newcommand{\codeCbarQC}[1]{\overline{\code{C}}_{\mathrm{QC}}^{(#1)}}

\newcommand{\matrHbar}{\overline{\matr{H}}}
\newcommand{\matrHbarconv}{\overline{\matr{H}}_{\mathrm{conv}}}
\newcommand{\matrHbarQC}[1]{\overline{\matr{H}}_{\mathrm{QC}}^{(#1)}}

\renewcommand{\leq}{\leqslant} 
\renewcommand{\geq}{\geqslant} 

\newcommand{\ms}{m_{\mathrm{s}}}
\newcommand{\Ts}{T_{\mathrm{s}}}
\newcommand{\nus}{\nu_{\mathrm{s}}}
 
\newcommand{\matrA}{\matr{A}}

\newcommand{\matrB}{\matr{B}}
\newcommand{\omatrB}{\overline{\matr{B}}}
\newcommand{\matrH}{\matr{H}}
\newcommand{\tmatrH}{\matr{\tilde H}} 
\newcommand{\matrI}{\matr{I}}
\newcommand{\matrM}{\matr{M}}
\newcommand{\matrP}{\matr{P}}
\newcommand{\matrT}{\matr{T}}

\newcommand{\matrzero}{\matr{0}}

\newcommand{\vv}{\vect{v}}
\newcommand{\ovv}{\overline{\vect{v}}}

\newcommand{\setI}{\set{I}}
\newcommand{\setL}{\set{L}}
\newcommand{\setM}{\set{M}}

\newcommand{\graph}[1]{\mathsf{#1}}
\newcommand{\graphT}{\graph{T}}
\newcommand{\tgraph}[1]{\mathsf{\tilde #1}}

\newcommand{\del}{\partial}

\newcommand{\GCCone}{\textbf{GCC1}}
\newcommand{\GCCtwo}{\textbf{GCC2}}

\newcommand{\mA}{m_{\matr{A}}}
\newcommand{\nA}{n_{\matr{A}}}

\newcommand{\R}{\mathbb{R}}

\newcounter{mytempeqcounter}

\newcommand{\bigformulatop}[2]{%
  \begin{figure*}[!t]
    \normalsize
    \setcounter{mytempeqcounter}{\value{equation}}
    \setcounter{equation}{#1}
    #2

    \setcounter{equation}{\value{mytempeqcounter}}
    \hrulefill
    \vspace*{4pt}
  \end{figure*}
}

\begin{document}

\title{Deriving Good LDPC Convolutional Codes from LDPC Block Codes}

\author{Ali E. Pusane,~\IEEEmembership{Member,~IEEE,} 
        Roxana Smarandache,~\IEEEmembership{Member,~IEEE,} \\
        Pascal O.~Vontobel,~\IEEEmembership{Member,~IEEE,} and 
        Daniel J.~Costello, Jr.,~\IEEEmembership{Life Fellow,~IEEE}%
        \thanks{Manuscript received April 25, 2010,
                revised August 17, 2010, and November 4, 2010.
          The first and fourth authors were partially
          supported by NSF Grants CCR-0205310 and CCF-0830650, and by NASA
          Grant NNX09AI66G. Additionally, the first author was supported by
          a Graduate Fellowship from the Center for Applied Mathematics,
          University of Notre Dame. The second author was supported by 
          NSF Grants DMS-0708033 and CCF-0830608, and partially supported by 
          NSF Grant CCR-0205310. Parts of the material in this paper were
          presented at the IEEE International Symposium on Information
          Theory, Nice, France, June 2007.}
        \thanks{A.~E.~Pusane was with the Department of Electrical
          Engineering, University of Notre Dame, Notre Dame, IN 46556, USA.
          He is now with the Department of Electrical and
          Electronics Engineering, Bogazici University, Bebek, Istanbul 34342,
          Turkey (email: ali.pusane@boun.edu.tr).}%
        \thanks{R.~Smarandache is with the Department of Mathematics and
          Statistics, San Diego State University, San Diego, CA 92182, USA
          (e-mail: rsmarand@sciences.sdsu.edu).}%
        \thanks{P.~O.~Vontobel is with Hewlett-Packard Laboratories, 1501 Page
          Mill Road, Palo Alto, CA 94304, USA
          (e-mail: pascal.vontobel@ieee.org).} %
        \thanks{D.~J.~Costello, Jr.\ is with the Department of Electrical
          Engineering, University of Notre Dame, Notre Dame, IN 46556, USA
          (e-mail: costello.2@nd.edu).}%
      }

\markboth{Accepted for IEEE Transactions on Information Theory}%
         {Pusane, Smarandache, Vontobel, Costello}

\maketitle

\begin{abstract}
  Low-density parity-check (LDPC) convolutional co\-des are capable of achieving
  excellent performance with low encoding and decoding complexity. In this
  paper we discuss several graph-cover-based methods for deriving families of
  time-invariant and time-varying LDPC convolutional codes from LDPC block
  codes and show how earlier proposed LDPC convolutional code constructions
  can be presented within this framework.

  Some of the constructed convolutional codes significantly outperform the
  underlying LDPC block codes. We investigate some possible reasons for this
  ``convolutional gain,'' and we also discuss the --- mostly moderate ---
  decoder cost increase that is incurred by going from LDPC block to LDPC
  convolutional codes.
\end{abstract}


\section{Introduction}

\IEEEPARstart{I}{n the last} fifteen years, the area of channel coding has
been revolutionized by the practical realization of capacity-approaching
coding schemes, initiated by the invention of turbo codes and their associated
decoding algorithms in 1993 \cite{Berrou:Glavieux:Thitimajshima:93:ICC}. A few
years after the invention of the turbo coding schemes, researchers became
aware that Gallager's low-density parity-check (LDPC) block codes and
message-passing iterative decoding, first introduced
in~\cite{Gallager:62:IRE}, were also capable of capacity-approaching
performance. The analysis and design of these coding schemes quickly attracted
considerable attention in the literature, beginning with the work of
Wiberg~\cite{Wiberg:96:diss}, MacKay and Neal \cite{Mackay:Neal:96:EL}, and
many others. An irregular version of LDPC codes was first introduced by Luby
\textit{et al.} in \cite{Luby:Mitzenmacher:Shokrollahi:Spielman:Stemann:97:1,
  Luby:Mitzenmacher:Shokrollahi:Spielman:01:IT}, and analytical tools were
presented in~\cite{Chung:Richardson:Urbanke:01:IT,Richardson:Urbanke:01:IT} to
obtain performance limits for graph-based message-passing iterative decoding
algorithms, such as those suggested by Tanner~\cite{Tanner:81:IT}. For many
classes of channels, these tools have been successfully employed to design
families of irregular LDPC codes that perform very well near
capacity~\cite{Chung:Forney:Richardson:Urbanke:01:COMML,
  Richardson:Shokrollahi:Urbanke:01:IT}. Moreover, for the binary erasure
channel these tools have enabled the design of families of irregular LDPC
codes that are not only capacity-approaching but in fact capacity-achieving
(see~\cite{Oswald:Shokrollahi:02:1} and references therein).

The convolutional counterparts of LDPC block codes are LDPC convolutional
codes. Analogous to LDPC block codes, LDPC convolutional codes are defined by
sparse parity-check matrices, which allow them to be decoded using iterative
message-passing algorithms. Recent studies have shown that LDPC convolutional
codes are suitable for practical implementation in a number of different
communication scenarios, including continuous transmission and block
transmission in frames of arbitrary
size~\cite{Bates:Elliot:Swamy:06:TCAS,Bates:Chen:Dong:05:Pacific,
  Bates:Gunthorpe:Pusane:Chen:Zigangirov:Costello:05:NASA}.

Two major methods have been proposed in the literature for the construction of
LDPC convolutional codes, two methods that in fact started the field of LDPC
convolutional codes. The first method was proposed by
Tanner~\cite{Tanner:81:Patent} (see also \cite{Tanner:87:TechnicalReport,
  Tanner:Sridhara:Sridharan:Fuja:Costello:04:IT}) and exploits similarities
between quasi-cyclic block codes and time-invariant convolutional codes. The
second method was presented by Jim\'enez-Feltstr\" om and Zigangirov
in~\cite{JimenezFeltstrom:Zigangirov:99:IT} and relies on a matrix-based
unwrapping procedure to obtain the parity-check matrix of a periodically
time-varying convolutional code from the parity-check matrix of a block code.

\bigformulatop{0}
{
  \begin{align}
    \matrHbarconv
      &= \begin{bmatrix}
        \matrH_0(0) \hfill & & & & & & & \\
        \matrH_1(1) \hfill & \matrH_0(1) \hfill & & & & & & \\
        \vdots & \vdots & \ddots & & & & & \\
        \matrH_{\ms}(\ms) \hfill & \matrH_{\ms-1}(\ms)\hfill & \cdots & 
        \matrH_0(\ms) \hfill & & & & \\
        & \matrH_{\ms}(\ms\!+\!1) \hfill & \matrH_{\ms-1}(\ms\!+\!1) \hfill & 
        \cdots & \matrH_0(\ms\!+\!1) \hfill & & & \\
        & \hskip2cm\ddots & \hskip2cm\ddots & & \hskip2cm\ddots & & & \\
        & & \matrH_{\ms}(t) \hfill & \matrH_{\ms-1}(t) \hfill & 
        \cdots & \matrH_0(t)\hfill & & \\
        & & \ddots & \ddots & & & \ddots
      \end{bmatrix}.
      \label{eq:PCC}
  \end{align}
}

\setcounter{equation}{1}

The aims of this paper are threefold. First, we show that these two LDPC
convolutional code construction methods, once suitably generalized, are in
fact tightly connected. We establish this connection with the help of
so-called graph covers.\footnote{Note that graph covers have been used in two
  different ways in the LDPC code literature. On the one hand, starting with
  the work of Tanner~\cite{Tanner:99:2}, they have been used to construct
  Tanner graphs~\cite{Tanner:81:IT} of LDPC codes, and therefore parity-check
  matrices of LDPC codes. Codes constructed in this way are nowadays often
  called proto-graph-based codes, following the influential work of
  Thorpe~\cite{Thorpe:03:Report}, who formalized this code construction
  approach. On the other hand, starting with the work of Koetter and
  Vontobel~\cite{Koetter:Vontobel:03:ISTC, Vontobel:Koetter:05:IT:subm},
  finite graph covers have been used to analyze the behavior of LDPC codes
  under message-passing iterative decoding. In this paper, we will use graph
  covers in the first way, with the exception of some comments on
  pseudo-codewords.}  A second aim is to discuss a variety of LDPC
convolutional code constructions. Although the underlying principles are
mathematically quite straightforward, it is important to understand how they
can be applied to obtain convolutional codes with good performance and
attractive encoder and decoder architectures. A third aim is to make progress
towards a better understanding of where the observed ``convolutional gain''
comes from, and what its costs are in terms of decoder complexity.

The paper is structured as follows. After some notational remarks in
Section~\ref{sec:notation:1:basic}, we discuss the basics of LDPC
convolutional codes in Section~\ref{sec:ldpc:convolutional:codes:1}. In
particular, in that section we give a first exposition of the LDPC
convolutional code construction methods due to Tanner and due to
Jim\'enez-Feltstr\" om and Zigangirov. In Section~\ref{sec:graph:covers:1} we
discuss two types of graph-cover code constructions and show how they can be
used to connect the code construction methods due to Tanner and due to
Jim\'enez-Feltstr\" om and Zigangirov. Based on these insights,
Section~\ref{sec:variations:unwrapping:1} presents a variety of LDPC
convolutional code constructions (along with simulation results), and in
Section~\ref{sec:other:LDPC:code:constructions:1} we mention some similarities
and differences of these constructions compared to other recent code
constructions in the literature. Afterwards, in Section~\ref{sec:analysis} we
analyze some aspects of the constructed LDPC convolutional codes and discuss
some possible reasons for the ``convolutional gain,'' before we conclude the
paper in Section~\ref{sec:discussion}.

\subsection{Notation}
\label{sec:notation:1:basic}

We use the following sets, rings, and fields: $\Z$ is the ring of integers;
$\Zp$ is the set of non-negative integers; $\Ftwo$ is the field of size two;
$\Ftwo[X]$ is the ring of polynomials with coefficients in $\Ftwo$ and
indeterminate $X$; and $\Ftwoxmodr$ is the ring of polynomials in $\Ftwo[X]$
modulo $X^r - 1$, where $r$ is a positive integer. We also use the notational
short-hand $\shortFtwoxmodr$ for $\Ftwoxmodr$.

By $\Ftwo^n$ and $\shortFtwoxmodr^n$,
we mean, respectively, a row vector over $\Ftwo$ of length $n$ and a row
vector over $\shortFtwoxmodr$ of length $n$.
In the following, if $\matrM$ is some matrix, then $[\matrM]_{j,i}$ denotes
the entry in the $j$-th row and $i$-th column of $\matrM$. Note that we use
the convention that indices of vector entries start at $0$ (and not at $1$),
with a similar convention for row and column indices of matrix entries. (This
comment applies also to semi-infinite matrices, which are defined such that
the row and column index sets equal $\Zp$.) The only exception to this
convention are bi-infinite matrices, where the row and column index sets equal
$\Z$.  Finally, $\matrM_1 \otimes \matrM_2$ will denote the Kronecker product
of the matrices $\matrM_1$ and $\matrM_2$.

\section{LDPC Convolutional Codes}
\label{sec:ldpc:convolutional:codes:1}

This section defines LDPC convolutional codes and discusses why they are
interesting from an implementation perspective. Afterwards, we review two
popular methods of obtaining LDPC convolutional codes by unwrapping block
codes. Later in this paper, namely in Section~\ref{sec:graph:covers:1}, we
will use graph covers to show how these two methods are connected, and in
Section~\ref{sec:variations:unwrapping:1} we will see how these two methods
can be implemented and combined to obtain LDPC convolutional codes with very
good performance.

\subsection{Definition of LDPC Convolutional Codes}
\label{sec:definition:ldpc:conv:codes:1}

A semi-infinite binary parity-check matrix as in~\eqref{eq:PCC} at the top of
this page defines a convolutional code $\codeCbarconv$ as follows. Namely, it
is the set of semi-infinite sequences given by
\begin{align*}
  \codeCbarconv
    &= \Bigl\{
         \ovv \in \GF{2}^{\infty}
         \Bigm| 
         \matrHbarconv \cdot \ovv^\tr = \vect{0}^\tr
       \Bigr\},
\end{align*}
where $(\,\cdot\,)^\tr$ denotes the transpose of a vector or of a matrix.

We comment on several important aspects and properties of the code
$\codeCbarconv$ and its parity-check matrix $\matrHbarconv$.
\begin{itemize}

\item If the submatrices $\matr{H}_{i}(t)$, $i = 0, 1, \cdots, \ms$, $t \in
  \Zp$, have size $(c-b) \times c$ with $b < c$, then $\codeCbarconv$ is said
  to have (design) rate $R = b/c$.

\item The parameter $\ms$ that appears in~\eqref{eq:PCC} is called the
  syndrome former memory. It is an important parameter of $\codeCbarconv$
  because the maximal number of non-zero submatrices per block row of
  $\matrHbarconv$ is upper bounded by $\ms + 1$.

\item The quantity $\nus = (\ms + 1) \cdot c$ is called the constraint length
  of $\codeCbarconv$. It measures the maximal width (in symbols) of the
  non-zero area of $\matrHbarconv$.\footnote{Strictly speaking, the above
    formula for $\nus$ gives only an upper bound on the maximal width (in
    symbols) of the non-zero area of $\matrHbarconv$, but this upper bound
    will be good enough for our purposes.}

\item We do not require that for a given $i = 0, 1, \ldots, \ms$ the
  submatrices $\{ \matrH_i(t) \}_{t \in \Zp}$ are independent of $t$, and so
  $\codeCbarconv$ is in general a \emph{time-varying} convolutional code.

\item If there is a positive integer $\Ts$ such that $\matr{H}_i(t) =
  \matr{H}_i(t + \Ts)$ for all $i = 0, 1, \ldots, \ms$ and all $t\in \Zp$,
  then $\Ts$ is called the period of $\matrHbarconv$, and $\codeCbarconv$ is
  \emph{periodically time-varying}.

\item If the period $\Ts$ equals $1$, then $\matrHbarconv$ is called
  \emph{time-invariant}, and the parity-check matrix can be simply written as
  \begin{align}
    \hskip-0.2cm
    \matrHbarconv
      &= \left[
         \begin{array}{@{\;}c@{\;\;}c@{\;\;}c@{\;\;}c@{\;\;}c@{\;}c@{\;}}
           \matrH_0 \hfill & & & & & \\
           \matrH_1\hfill & \matrH_0 \hfill & & & & \\
           \vdots & \vdots & \ddots & & & \\
           \matrH_{\ms} \hfill & \matrH_{\ms-1} \hfill & \ldots & 
           \matrH_0 \hfill & & \\
           & \matrH_{\ms} \hfill & \matrH_{\ms-1} \hfill & \ldots &
           \matrH_0 \hfill & \\
           & \ddots & \ddots & & & \ddots
         \end{array}
         \right].
      \label{eq:TIPCC}
  \end{align}

\item If the number of ones in each row and column of $\matrHbarconv$ is small
  compared to the constraint length $\nus$, then $\codeCbarconv$ is an LDPC
  convolutional code.

\item An LDPC convolutional code $\codeCbarconv$ is called
  $(m_{\textrm{s}},J,K)$-regular if, starting from the zeroth column,
  $\matrHbarconv$ has $J$ ones in each column, and, starting from the
  $(m_{\mathrm{s}} + 1) \cdot (c-b)$-th row, $\matrHbarconv$ has $K$ ones in
  each row. If, however, there are no integers $\ms$, $J$, and $K$ such that
  $\codeCbarconv$ is $(m_{\textrm{s}},J,K)$-regular, then $\codeCbarconv$ is
  called irregular.

\end{itemize}

Of course, there is some ambiguity in the above definition. Namely, a
periodically time-varying LDPC convolutional code with parameters $\Ts$, $b$,
and $c$ can also be considered to be a periodically time-varying LDPC
convolutional code with parameters $\Ts' = \Ts/\ell$, $b' = \ell \cdot b$, $c'
= \ell \cdot c$, and $R' = b'/c' = b/c$ for any integer $\ell$ that divides
$\Ts$. In particular, for $\ell = \Ts$ we consider the code to be a
time-invariant LDPC convolutional code with parameters $b' = \Ts \cdot b$ and
$c' = \Ts \cdot c$.

\subsection{Implementation Aspects of LDPC Convolutional Codes}
\label{sec:implement}

\begin{figure*}
  \begin{center}
    \epsfig{file=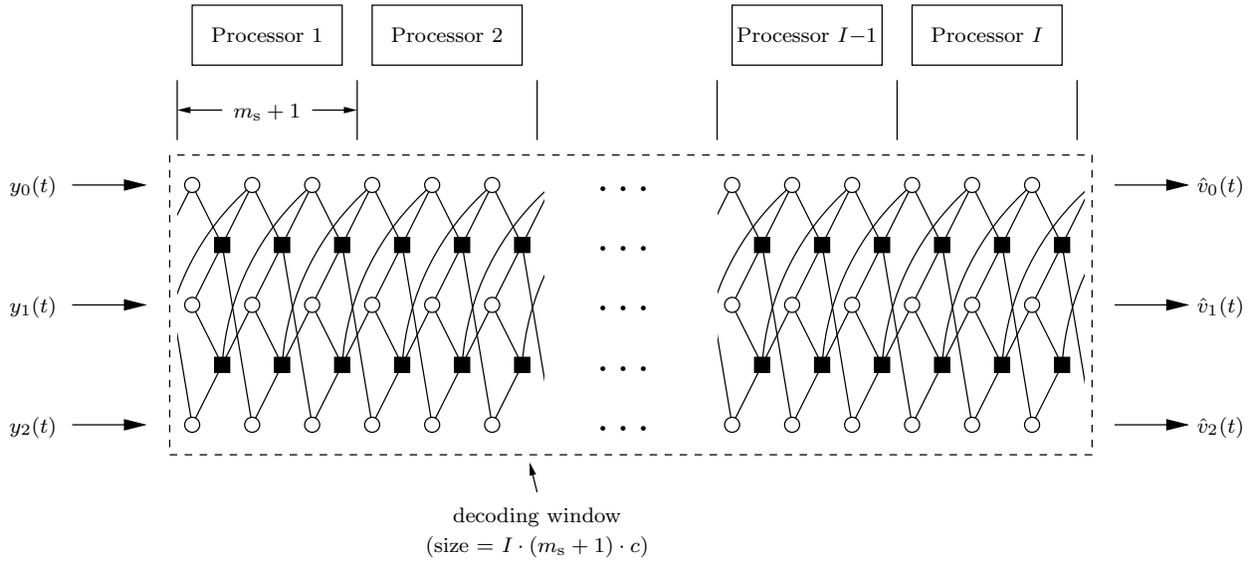, width=0.9\linewidth}
  \end{center}
  \caption{Tanner graph of a rate-$1/3$ convolutional code and an illustration
    of pipeline decoding. Here, $y_0(t)$, $y_1(t)$, $y_2(t)$ denote the stream
    of channel output symbols, and $\hat v_0(t)$, $\hat v_1(t)$, $\hat v_2(t)$
    denote the stream of decoder code bit decisions.}
  \label{fig:decoder}
\end{figure*}

An advantage of LDPC convolutional codes compared to their block code
counterparts is the so-called ``fast encoding'' property. As a result of the
diagonal shape of their parity-check matrices, many LDPC convolutional codes
enjoy simple shift register based encoders. Even randomly constructed LDPC
convolutional codes can be formed in such a way as to achieve this feature
without any loss in performance (see, e.g.,
\cite{JimenezFeltstrom:Zigangirov:99:IT,Pusane:Zigangirov:Costello:06:ICC}).
On the other hand, in order to have a simple encoding procedure for LDPC block
codes, either the block code must have some sort of structure
\cite{Zongwang:Lei:Lingqi:Lin:Fong:06:TCOM} or the parity-check matrix must be
changed to a more easily ``encodable''
form~\cite{Richardson:Urbanke:01:IT:encoding}.

The difficulty in constructing and decoding LDPC convolutional codes is
dealing with the unbounded size of the parity-check matrix. This is overcome
at the code construction step by considering only periodically time-varying or
time-invariant codes. The code construction problem is therefore reduced to
designing just one period of the parity-check matrix. For decoding, the most
obvious approach is to terminate the encoder and to employ message-passing
iterative decoding based on the complete Tanner graph representation of the
parity-check matrix of the code. Although this would be straightforward to
implement using a standard LDPC block code decoder, it would be wasteful of
resources, since the resulting (very large) block decoder would not be taking
advantage of two important aspects of the convolutional structure: namely,
that decoding can be done continuously without waiting for an entire
terminated block to be received and that the distance between two variable
nodes that are connected to the same check node is limited by the size of the
syndrome former memory.

In order to take advantage of the convolutional nature of the parity-check
matrix, a continuous sliding window message-passing iterative decoder that
operates on a window of size $I\cdot \nus$ variable nodes, where $I$ is the
number of decoding iterations to be performed, can be implemented, similar to
a Viterbi decoder with finite path memory~\cite{Lin:Costello:04:Book}. This
window size is chosen since, in a single iteration, messages from variable (or
check) nodes can be passed across a span of only one constraint length. Thus,
in $I$ iterations, messages can propagate only over a window of size $I$
constraint length. (See also the recent paper by Papaleo \textit{et
  al.}~\cite{Papaleo:Iyengar:Siegel:Wolf:Corazza:10:ITW}, which investigates
further reducing the window size for codes operating on a binary erasure
channel (BEC).) Another simplification is achieved by exploiting the fact that
a single decoding iteration of two variable nodes that are at least $\ms + 1$
time units apart can be performed independently, since the corresponding bits
cannot participate in the same parity-check equation.  This allows the
parallelization of the $I$ iterations by employing $I$ {\em independent
  identical processors} working on different regions of the parity-check
matrix simultaneously, resulting in the parallel pipeline decoding
architecture introduced in~\cite{JimenezFeltstrom:Zigangirov:99:IT}. The
pipeline decoder outputs a continuous stream of decoded data after an initial
decoding delay of $I \cdot \nus$ received symbols. The operation of this
decoder on the Tanner graph of a simple time-invariant rate-$1/3$
convolutional code with $\ms = 2$ and $\nus = 9$ is illustrated in
Figure~\ref{fig:decoder}.\footnote{For LDPC convolutional codes the parameter
  $\nus$ is usually much larger than typical values of $\nus$ for
  ``classical'' convolutional codes. Therefore the value $\nus = 9$ of the
  convolutional code shown in Figure~\ref{fig:decoder} is not typical for the
  codes considered in this paper.}

Although the pipeline decoder is capable of fully parallelizing the iterations
by using $I$ independent identical processors, employing a large number of
hardware processors might not be desirable in some applications. In such
cases, fewer processors (even one processor) can be scheduled to perform
subsets of iterations, resulting in a serial looping architecture
\cite{Bates:Chen:Gunthorpe:Pusane:Zigangirov:Costello:08:TCAS} with reduced
throughput. This ability to balance the processor load and decoding speed
dynamically is especially desirable where very large LDPC convolutional codes
must be decoded with limited available on-chip memory. Further discussion on
the implementation aspects of the pipeline decoder can be found in
\cite{Pusane:Jimenez:Sridharan:Lentmaier:Zigangirov:Costello:08:TCOM}.

\subsection{Unwrapping Techniques due to Tanner 
                    and due to Jim\'enez-Feltstr\"om and Zigangirov (JFZ)}
\label{sec:unwrapping:1}

In this subsection we discuss two approaches for deriving convolutional codes
from block codes, in particular for deriving LDPC convolutional codes from
LDPC block codes. The first technique will be the unwrapping due to Tanner and
the second will be the unwrapping due to Jim\'enez-Feltstr\"om and Zigangirov
(JFZ). In Section~\ref{sec:graph:covers:1} we will see, with the help of graph
covers, how these two -- seemingly different -- unwrapping techniques are
connected with each other.

The term {\em unwrapping}, in particular unwrapping a quasi-cyclic block code
to obtain a time-invariant convolutional code, was first introduced in a paper
by Tanner~\cite{Tanner:87:TechnicalReport} (see
also~\cite{Tanner:81:Patent}). That paper describes a link between
quasi-cyclic block codes and time-invariant convolutional codes and shows that
the free distance of the unwrapped convolutional code cannot be smaller than
the minimum distance of the underlying quasi-cyclic code. This idea was later
extended in~\cite{Levy:Costello:93:BookChapter,
  Esmaeili:Gulliver:Secord:Mahmoud:98:IT}.
  
Consider the quasi-cyclic block code $\codeCQC{r}$ defined by the polynomial
parity-check matrix $\matrHQC{r}(X)$ of size $m \times n$, i.e.,
\begin{align*}
  \codeCQC{r}
    &= \Bigl\{
         \vv(X) \in \shortFtwoxmodr^n
         \Bigm| 
         \matrHQC{r}(X) \cdot \vv(X)^\tr = \vect{0}^\tr
       \Bigr\}.
\end{align*}
Here the polynomial operations are performed modulo $X^r - 1$. The Tanner
unwrapping technique is simply based on dropping these modulo
computations. More precisely, with a quasi-cyclic block code $\codeCQC{r}$ we
associate the convolutional code
\begin{align*}
  \codeCconv
    &= \Bigl\{
         \vv(D) \in \GF{2}[D]^n
         \Bigm|
         \matrHconv(D) \cdot \vv(D)^\tr = \vect{0}^\tr
       \Bigr\}
\end{align*}
with polynomial parity-check matrix
\begin{align}
  \matrHconv(D)
    &\defeq
       \left.
         \matrHQC{r}(X)
       \right|_{X = D}.
         \label{eq:from:X:to:D:1}
\end{align}
Here the change of indeterminate from $X$ to $D$ indicates the lack of the
modulo $D^r-1$ operations. (Note that in~\eqref{eq:from:X:to:D:1} we assume
that the exponents appearing in the polynomials in $\matrHQC{r}(X)$ are
between $0$ and $r-1$ inclusive.)

It can easily be seen that any codeword $\vv(D)$ in $\codeCconv$ maps to a
codeword $\vv(X)$ in $\codeCQC{r}$ through
\begin{align*}
  \vect{v}(X)
    &\defeq
       \Big.
         \vect{v}(D) \ \mathrm{mod} \ (D^r-1)
       \Big|_{D = X},
\end{align*}
a process which was described in~\cite{Tanner:87:TechnicalReport} as the
wrapping around of a codeword in the convolutional code into a codeword in the
quasi-cyclic code. The inverse process was described as unwrapping.

Having introduced the unwrapping technique due to Tanner, we move on to
discuss the unwrapping technique due to
JFZ~\cite{JimenezFeltstrom:Zigangirov:99:IT}, which is another way to unwrap a
block code to obtain a convolutional code. The basic idea is best explained
with the help of an example.

\begin{figure} 
  \begin{center}
    \epsfig{file=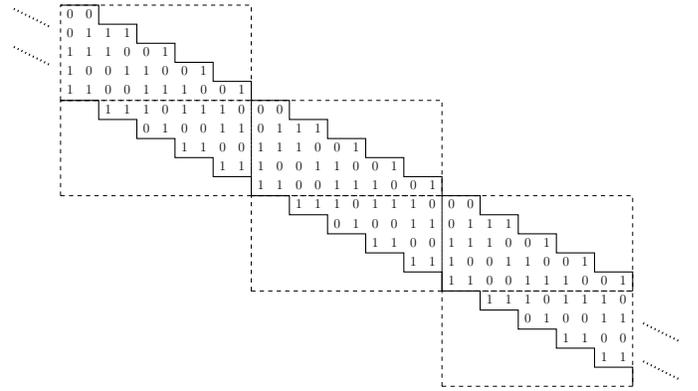, width=1.0\columnwidth}
  \end{center}
  \caption{Deriving a rate $R = 1/2$ periodically time-varying convolutional
    code with $b = 1$, $c = 2$, $\ms = 4$, $\nus = 10$, and $\Ts = 5$ from a
    rate-$1/2$ block code of length $10$.}
  \label{fig:jjj:1} 
\end{figure}

\begin{example}
  \label{example:unwrap:1}
  
  Consider the parity-check matrix
  \begin{center}
    \epsfig{file=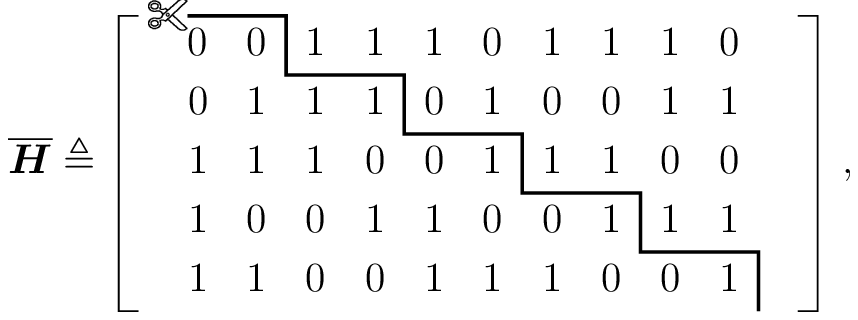, scale=0.85}
  \end{center}
  with size $m \times n = 5 \times 10$, of a rate-$1/2$ block code. As
  indicated above, we can take a pair of scissors and ``cut'' the parity-check
  matrix into two pieces, whereby the cutting pattern is such that we
  repeatedly move $c = 2$ units to the right and then $c-b = 1$ unit
  down. Having applied this ``diagonal cut,'' we repeatedly copy and paste the
  two parts to obtain the bi-infinite matrix shown in
  Figure~\ref{fig:jjj:1}. This new matrix can be seen as the parity-check
  matrix of (in general) a periodically time-varying convolutional code (here
  the period is $\Ts = 5$). It is worth observing that this new matrix has the
  same row and column weights as the matrix that we started with.\footnote{In
    practice, the codewords start at some time, so the convolutional
    parity-check matrix has effectively the semi-infinite form
    of~\eqref{eq:PCC}, and the row weights of the first $\nus - 1$ rows are
    reduced.} \exend
\end{example}

This example can be formalized easily. Namely, starting with an $m \times n$
parity-check matrix $\matrHbar$ of some block code, let $\eta \defeq \gcd(m,
n)$. Then the ``diagonal cut'' is performed by alternately moving $c = n /
\eta$ units to the right and then $c - b \defeq m/\eta$ units down (i.e., $b =
\bigl( (n-m)/\eta \bigr)$. The resulting convolutional code has rate $R=b/c$,
syndrome former memory $\ms = \eta - 1$, constraint length $\nus = (\ms+1)
\cdot c = \eta \cdot c = n$, and period $\Ts = \ms + 1 = \eta$.

Analogous to the comment at the end of
Section~\ref{sec:definition:ldpc:conv:codes:1}, it is also possible to cut the
matrix $\matrHbar$ in larger step sizes, e.g., moving $c' = \ell \cdot c$
units to the right and $c' - b' = \ell \cdot (c - b)$ units down, for any
integer $\ell$ that divides $\Ts = \eta$, thereby obtaining a periodically
time-varying convolutional code with rate $R' = b' / c' = b / c$, syndrome
former memory $\ms' = (\eta/\ell) - 1$, constraint length $\nus' = (\ms' + 1)
\cdot c' = \eta \cdot c = n$, and period $\Ts' = \ms' + 1 = \eta / \ell$. (See
also the discussion in
Section~\ref{sec:time:varying:ldpc:convolutional:codes:1}.)

In the rest of this paper, the term ``JFZ unwrapping technique'' will also
stand for the following generalization of the above procedure. Namely,
starting with a length-$n$ block code $\codeCbar$ defined by some size-$m
\times n$ parity-check matrix $\matrHbar$, i.e.,
\begin{align*}
  \codeCbar
    &= \Bigl\{
         \ovv \in \GF{2}^n
         \Bigm| 
         \matrHbar \cdot \ovv^\tr = \vect{0}^\tr
       \Bigr\},
\end{align*}
we write $\matrHbar$ as the sum $\matrHbar = \sum_{\ell \in \setL}
\matrH_{\ell} \ (\text{in $\Z$})$ of a collection of matrices $\{
\matrH_{\ell} \}_{\ell \in \setL}$. The convolutional code $\codeCbarconv$ is
then defined to be
\begin{align}
  \codeCbarconv
    &\defeq
       \Bigl\{
         \ovv \in \GF{2}^{\infty}
         \Bigm| 
         \matrHbarconv \cdot \ovv^\tr = \vect{0}^\tr
       \Bigr\},
         \label{eq:code:C:conv:JFZ:unwrapping:1}
\end{align}
where
\begin{align*} 
  \matrHbarconv
    &\defeq
      \begin{bmatrix}
         \matr{H}_0 \hfill & & & & & \\
         \matr{H}_1 \hfill & \matr{H}_0 \hfill & & & & \\
         \vdots & \vdots & \ddots & & & \\
         \matr{H}_{|\setL|-1} \hfill & \matr{H}_{|\setL|-2} \hfill & 
         \ldots & \matr{H}_0 \hfill & & \\
         & \matr{H}_{|\setL|-1} \hfill & \matr{H}_{|\setL|-2} \hfill & 
         \ldots & \matr{H}_{0} \hfill & \\
         & & \ddots & \ddots & \ddots & \ddots 
       \end{bmatrix}.
\end{align*}
Referring to the notation introduced in
Section~\ref{sec:definition:ldpc:conv:codes:1}, the matrix $\matrHbarconv$ is
the parity-check matrix of a time-invariant convolutional code. However,
depending on the decomposition of $\matrHbar$ and the internal structure of
the terms in that decomposition, the matrix $\matrHbarconv$ can also be (and
very often is) viewed as the parity-check matrix of a time-varying
convolutional code with non-trivial period $\Ts$.

In order to illustrate the generalization of the JFZ unwrapping technique that
we have introduced in the last paragraph, observe that decomposing $\matrHbar$
from Example~\ref{example:unwrap:1} as $\matrHbar = \matrH_0 + \matrH_1 \
(\text{in $\Z$})$ with
\begin{center}
  \epsfig{file=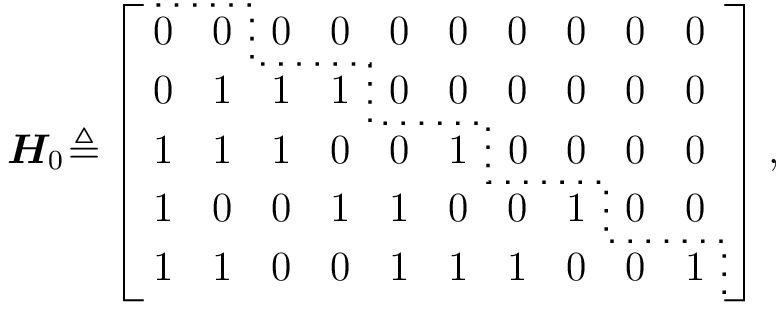, scale=0.85}
\end{center}
\begin{center}
  \epsfig{file=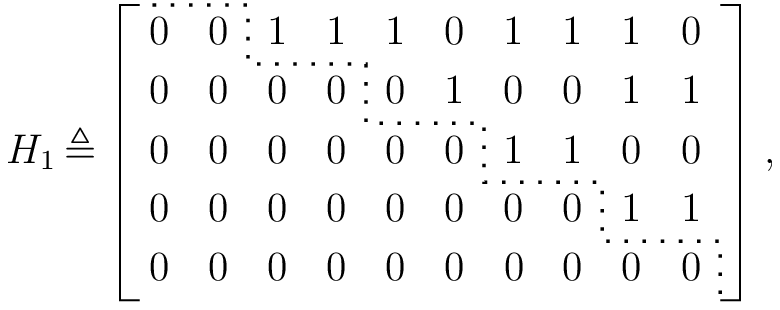, scale=0.85}
\end{center}
yields a convolutional code with parity-check matrix $\matrHbarconv$ whose
bi-infinite version equals the matrix shown in Figure~\ref{fig:jjj:1}.

\section{Tanner Graphs from Graph Covers}
\label{sec:graph:covers:1}

Having formally introduced LDPC convolutional codes in the previous section,
we now turn our attention to the main tool of this paper, namely graph
covers.

\begin{definition}[see, e.g., \cite{Stark:Terras:96:1}]
  \label{def:graph:cover:1}

  A {\em cover} of a graph $\graph{G}$ with vertex set $\set{W}$ and edge set
  $\set{E}$ is a graph $\tgraph{G}$ with vertex set $\set{\tilde W}$ and edge
  set $\set{\tilde E}$, along with a surjection $\pi: \set{\tilde W} \to
  \set{W}$ which is a graph homomorphism (i.e., $\pi$ takes adjacent vertices
  of $\tgraph{G}$ to adjacent vertices of $\graph{G}$) such that for each
  vertex $w \in \set{W}$ and each $\tilde w \in \pi^{-1}(w)$, the neighborhood
  $\del(\tilde w)$ of $\tilde w$ is mapped bijectively to $\del(w)$. A cover
  is called an {\em $M$-cover}, where $M$ is a positive integer, if $\bigl|
  \pi^{-1}(w) \bigr| = M$ for every vertex $w$ in $\set{W}$.\footnote{The
    number $M$ is also known as the degree of the cover. (Not to be confused
    with the degree of a vertex.)} \defend
\end{definition}

These graph covers will be used for the construction of new Tanner
graphs from old Tanner graphs, in particular for the construction of Tanner
graphs that represent LDPC convolutional codes.

More specifically, this section starts by discussing two simple methods to
specify a graph cover, which will be called graph-cover construction~1
(\GCCone) and graph-cover construction~2 (\GCCtwo). Although they yield
isomorphic Tanner graphs, and therefore equivalent codes, it is convenient to
have both methods at hand.\footnote{For a formal definition of code
  equivalence, see for example~\cite{MacWilliams:Sloane:77}.} As we will see,
interesting classes of Tanner graphs can be obtained by repeatedly applying
these graph-cover constructions, by mixing them, and by suitably shortening
the resulting codes. Moreover, these two graph-cover constructions will allow
us to exhibit a connection between the Tanner and the JFZ unwrapping
techniques.

\subsection{Graph-Cover Constructions}

Let $\matrA$ be an $\mA \times \nA$ matrix over $\Zp$. With such a matrix we
can associate a Tanner graph $\graphT(\matrA)$, where we draw $\nA$ variable
nodes, $\mA$ check nodes, and where there are $[\matrA]_{j,i}$ edges from the
$i$-th variable node to the $j$-th check node.\footnote{Note that we use a
  generalized notion of Tanner graphs, where parallel edges are allowed and
  are reflected by corresponding integer entries in the associated matrix.}
Given the role that the matrix $\matrA$ will play subsequently, we
follow~\cite{Thorpe:03:Report} and call the matrix $\matrA$ a proto-matrix and
the corresponding graph $\graphT(\matrA)$ a proto-graph.

The next definition introduces \GCCone\ and \GCCtwo, two ways to specify graph
covers that will be used throughout the rest of the paper.\footnote{We leave
  it as an exercise for the reader to show that the graphs constructed in
  \GCCone\ and \GCCtwo\ are indeed two instances of the graph cover definition
  in Definition~\ref{def:graph:cover:1}.}

\begin{definition}
  \label{def:graph:cover:construction:1}

  For some positive integers $\mA$ and $\nA$, let $\matrA \in \Zp^{\mA \times
    \nA}$ be a proto-matrix. We also introduce the following objects:
  \begin{itemize}
   
  \item For some finite set $\setL$, let $\{ \matrA_{\ell} \}_{\ell \in
      \setL}$ be a collection of matrices such that $\matrA_{\ell} \in
    \Zp^{\mA \times \nA}$, $\ell \in \setL$, and such that $\matrA =
    \sum_{\ell \in \setL} \matr{A}_{\ell} \text{ (in
      $\Z$)}$.

  \item For some positive integer $r$, let $\{ \matrP_{\ell} \}_{\ell \in
      \setL}$ be a collection of size-$r {\times} r$ permutation matrices.
    I.e., for every $\ell \in \setL$, the matrix $\matrP_{\ell}$ is such that
    it contains one ``$1$'' per column, one ``$1$'' per row, and ``$0$''s
    otherwise.

  \end{itemize}
  Based on the collection of matrices $\{ \matrA_{\ell} \}_{\ell \in \setL}$
  and the collection of matrices $\{ \matrP_{\ell} \}_{\ell \in \setL}$, there
  are two common ways of defining a graph cover of the Tanner graph
  $\graphT(\matrA)$. (In the following expressions, $\matrI$ is the identity
  matrix of size $r \times r$.)
  \begin{itemize}
  
  \item \textbf{Graph-cover construction 1 (\GCCone).} Consider the
    intermediary matrix
    \begin{align*}
      \matrB'
        &\defeq
           \matrA
           \otimes
           \matrI
         = \sum_{\ell \in \setL}
             \left(
               \matrA_{\ell}
               \otimes
               \matrI
             \right),
    \end{align*}
    whose Tanner graph $\graphT(\matrB')$ consists of $r$ disconnected copies
    of $\graphT(\matrA)$. This is an $r$-fold cover of $\graphT(\matrA)$,
    albeit a rather trivial one. In order to obtain an interesting $r$-fold
    graph cover of $\matrA$, for each $\ell \in \setL$, we replace
    $\matrA_{\ell} \otimes \matrI$ by $\matrA_{\ell} \otimes \matrP_{\ell}$,
    i.e., we define
    \begin{align*}
      \matrB
        &\defeq
           \sum_{\ell \in \setL}
             \left(
               \matrA_{\ell}
               \otimes
               \matrP_{\ell}
             \right).
    \end{align*}

  \item \textbf{Graph-cover construction 2 (\GCCtwo)} Consider the
    intermediary matrix
    \begin{align*}
      \omatrB'
        &\defeq
           \matrI
           \otimes
           \matrA
         = \sum_{\ell \in \setL}
             \left(
               \matrI
               \otimes
               \matrA_{\ell}
             \right),
    \end{align*}
    whose Tanner graph $\graphT(\omatrB')$ consists of $r$ disconnected copies
    of $\graphT(\matrA)$. This is an $r$-fold cover of $\graphT(\matrA)$,
    albeit a rather trivial one. In order to obtain an interesting $r$-fold
    graph cover of $\matrA$, for each $\ell \in \setL$, we replace $\matrI
    \otimes \matrA_{\ell} $ by $\matrP_{\ell} \otimes \matrA_{\ell}$, i.e., we
    define
    \begin{align*}
      \omatrB
        &\defeq
           \sum_{\ell \in \setL}
             \left(
               \matrP_{\ell}
               \otimes
               \matrA_{\ell}
             \right).
    \end{align*}

  \end{itemize}
  If all the matrices $\{ \matrP_{\ell} \}_{\ell \in \setL}$ are circulant
  matrices, then the graph covers $\graphT(\matrB)$ and $\graphT(\omatrB)$
  will be called cyclic covers of $\graphT(\matrA)$. \defend
\end{definition}

One can verify that the two graph-cover constructions in
Definition~\ref{def:graph:cover:construction:1} are such that the matrix
$\matrB$, after a suitable reordering of the rows and columns, equals the
matrix $\omatrB$.\footnote{Indeed, a possible approach to show this is to use
  the fact that $\matrA_{\ell} \otimes \matrP_{\ell}$ and $\matrP_{\ell}
  \otimes \matrA_{\ell}$ are permutation equivalent, i.e., there is a pair of
  permutation matrices $(\matr{Q}, \matr{Q}')$ such that $\matrA_{\ell}
  \otimes \matrP_{\ell} = \matr{Q} \cdot (\matrP_{\ell} \otimes \matrA_{\ell})
  \cdot \matr{Q}'$. Of course, for this to work, the pair $(\matr{Q},
  \matr{Q}')$ must be independent of $\ell \in \setL$, i.e., dependent only on
  the size of the matrices $\{ \matrA_{\ell} \}_{\ell \in \setL}$ and $\{
  \matrP_{\ell} \}_{\ell \in \setL}$. Such a $(\matr{Q}, \matr{Q}')$ pair can
  easily be found.}  This implies that $\graphT(\matrB)$ and
$\graphT(\omatrB)$ are isomorphic graphs; nevertheless, it is helpful to
define both types of constructions.

\subsection{Graph-Cover Construction Special Cases}

The following examples will help us to better understand how \GCCone\ and
\GCCtwo\ can be used to obtain interesting classes of Tanner graphs, and, in
particular, how the resulting graph-cover constructions can be visualized
graphically. Although these examples are very concrete, they are written such
that they can be easily generalized.

\begin{example}[Cyclic cover]
  \label{ex:graph:cover:construction:1}

  Consider the proto-matrix
  \begin{align}
    \matr{A}
      \defeq
        \begin{bmatrix}
          1 & 1 & 1 \\
          1 & 1 & 1
        \end{bmatrix}
          \label{eq:ex:graph:cover:constructions:1:A:1}
  \end{align}
  with $\mA = 2$ and $\nA = 3$, and whose Tanner graph $\graphT(\matrA)$ is
  shown in Figure~\ref{fig:graph:cover:constructions:1}(a). Let $\setL \defeq
  \{ 0, 1 \} \times \{ 0, 1, 2\}$, and let the collection of matrices $\{
  \matrA_{\ell} \}_{\ell \in \setL}$ be given by $\{ \matrA_{j,i} \}_{j,i}$,
  where for each $j = 0, \ldots, \mA-1$ and each $i = 0, \ldots, \nA-1$ the
  matrix $\matrA_{j,i} \in \Zp^{\mA \times \nA}$ is defined as follows
  \begin{align*}
    [\matrA_{j,i}]_{j',i'}
      &\defeq
         \begin{cases}
           [\matrA]_{j',i'} & \text{if $(j',i') = (j,i)$} \\
           0                & \text{otherwise}
         \end{cases}.
  \end{align*}
  Moreover, let $r \defeq 7$, and let the collection of matrices $\{
  \matrP_{\ell} \}_{\ell \in \setL}$ be given by $\{ \matrP_{j,i} \}_{j,i}$,
  where $\matrP_{0,0} \defeq \matrI_1$, $\matrP_{0,1} \defeq \matrI_2$,
  $\matrP_{0,2} \defeq \matrI_4$, $\matrP_{1,0} \defeq \matrI_6$,
  $\matrP_{1,1} \defeq \matrI_5$, $\matrP_{1,2} \defeq \matrI_3$, and where
  $\matrI_s$ is an $s$ times left-shifted identity matrix of size $r \times
  r$.
  \begin{itemize}

  \item Using \GCCone, we obtain the matrices
    \begin{align}
      \matrB'
        &= \matrA \otimes \matrI_0 
         = \begin{bmatrix}
             \matrI_0 & \matrI_0 & \matrI_0 \\
             \matrI_0 & \matrI_0 & \matrI_0
           \end{bmatrix}, \nonumber \\[0.25cm]
      \matrB
        &= \begin{bmatrix}
             \matrI_1 & \matrI_2 & \matrI_4 \\
             \matrI_6 & \matrI_5 & \matrI_3
           \end{bmatrix},
             \label{eq:ex:graph:cover:constructions:1:B:1}
    \end{align}
    whose Tanner graphs $\graphT(\matrB')$ and $\graphT(\matrB)$, respectively,
    are shown in Figure~\ref{fig:graph:cover:constructions:1}(b).

  \item Using \GCCtwo, we obtain the matrices
    \begin{align}
      \hskip-0.20cm
      \omatrB'
        &= \matrI_0 \otimes \matrA
         = \begin{bmatrix}
             \matrA   & \matr{0} & \matr{0} & \matr{0} &
             \matr{0} & \matr{0} & \matr{0} \\
             \matr{0} & \matrA   & \matr{0} & \matr{0} &
             \matr{0} & \matr{0} & \matr{0} \\
             \matr{0} & \matr{0} & \matrA   & \matr{0} &
             \matr{0} & \matr{0} & \matr{0} \\
             \matr{0} & \matr{0} & \matr{0} & \matrA &
             \matr{0} & \matr{0} & \matr{0} \\
             \matr{0} & \matr{0} & \matr{0} & \matr{0} &
             \matrA   & \matr{0} & \matr{0} \\
             \matr{0} & \matr{0} & \matr{0} & \matr{0} &
             \matr{0} & \matrA   & \matr{0} \\
             \matr{0} & \matr{0} & \matr{0} & \matr{0} &
             \matr{0} & \matr{0} & \matrA
           \end{bmatrix}\!\!, \nonumber \\[0.25cm]
      \hskip-0.20cm
      \omatrB
        &= \left[
           \begin{array}{@{\;}c@{\;\;}c@{\;\;}c@{\;\;}c@{\;\;}c@{\;\;}c@{\;\;}c@{\;}}
             \matr{0} & \matrA_{1,0} & \matrA_{1,1} & \matrA_{0,2} &
             \matrA_{1,2} & \matrA_{0,1} & \matrA_{0,0} \\
             \matrA_{0,0} & \matr{0} & \matrA_{1,0} & \matrA_{1,1} &
             \matrA_{0,2} & \matrA_{1,2} & \matrA_{0,1} \\
             \matrA_{0,1} & \matrA_{0,0} & \matr{0} & \matrA_{1,0} &
             \matrA_{1,1} & \matrA_{0,2} & \matrA_{1,2} \\
             \matrA_{1,2} & \matrA_{0,1} & \matrA_{0,0} & \matr{0} &
             \matrA_{1,0} & \matrA_{1,1} & \matrA_{0,2} \\
             \matrA_{0,2} & \matrA_{1,2} & \matrA_{0,1} & \matrA_{0,0} &
             \matr{0} & \matrA_{1,0} & \matrA_{1,1} \\
             \matrA_{1,1} & \matrA_{0,2} & \matrA_{1,2} & \matrA_{0,1} &
             \matrA_{0,0} & \matr{0} & \matrA_{1,0} \\
             \matrA_{1,0} & \matrA_{1,1} & \matrA_{0,2} & \matrA_{1,2} &
             \matrA_{0,1} & \matrA_{0,0} & \matr{0}
           \end{array}
           \right]\!\!,
             \label{eq:ex:graph:cover:constructions:1:oB:1}
    \end{align}
    whose Tanner graphs $\graphT(\omatrB')$ and $\graphT(\omatrB)$,
    respectively, are shown in
    Figure~\ref{fig:graph:cover:constructions:1}(c). Note that all the block
    rows and all the block columns sum (in $\Z$) to $\matrA$. (This
    observation holds in general, not just for this example.) \exend

  \end{itemize}
\end{example}

\begin{figure}
  \begin{center}
    \subfigure[Proto-graph $\graphT(\matrA)$.]
    {
       \epsfig{file=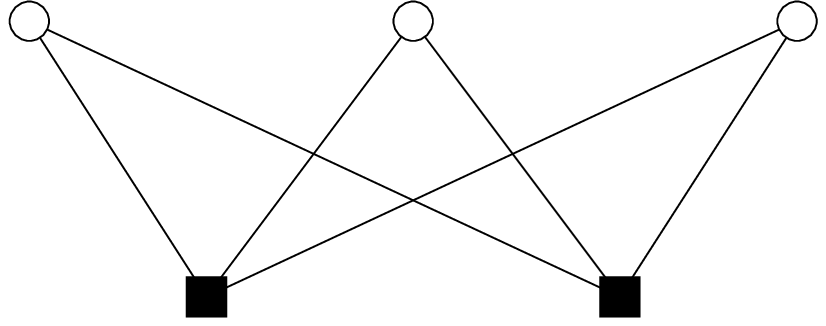, scale=0.7}
    }

    \bigskip

    \subfigure[\GCCone\ based on $\graphT(\matrA)$.
               Top: $\graphT(\matrB')$.
               Bottom: $\graphT(\matrB)$.]
    {
      \epsfig{file=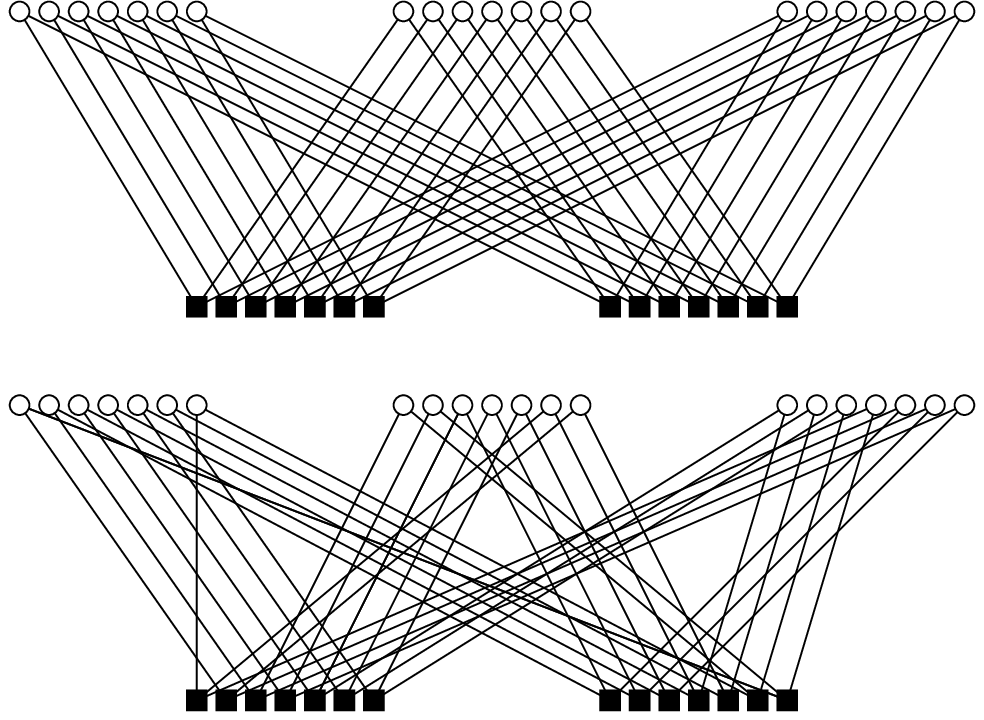, scale=0.7}
    }

    \bigskip

    \subfigure[\GCCtwo\ based on $\graphT(\matrA)$.
               Top: $\graphT(\omatrB')$.
               Bottom: $\graphT(\omatrB)$.]
    {
      \epsfig{file=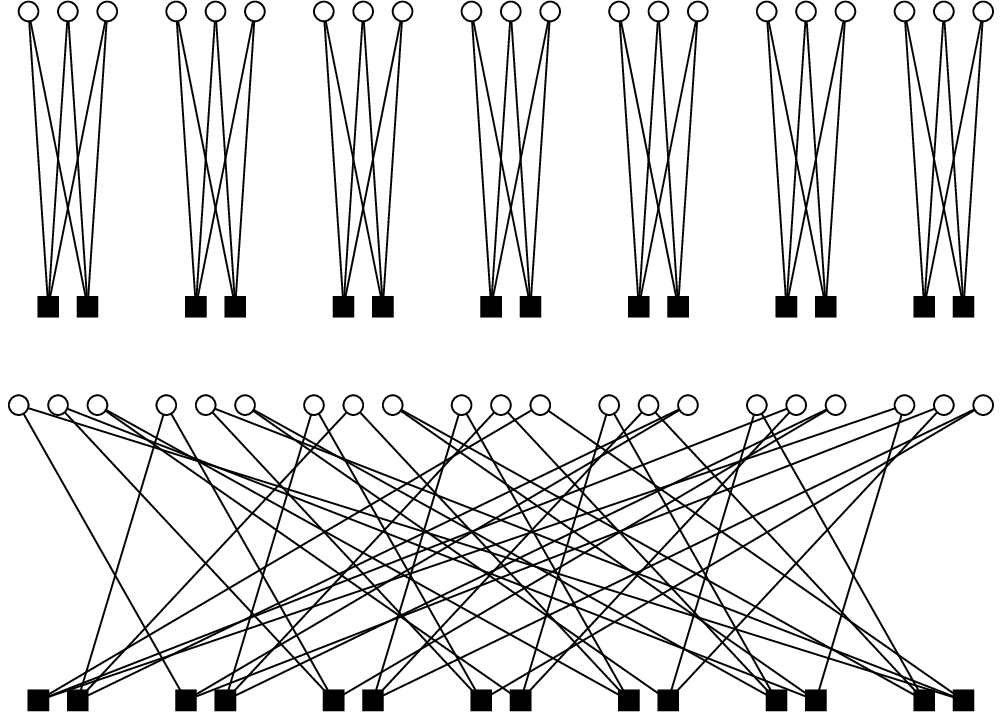, scale=0.7} 
    }
  \end{center}
  \caption{Proto-graph and graph-covers for the graph-cover constructions
    discussed in Example~\ref{ex:graph:cover:construction:1}. (Compare with
    the corresponding graphs in
    Figure~\ref{fig:graph:cover:constructions:1:variation:2}.)}
  \label{fig:graph:cover:constructions:1}
\end{figure}

\begin{figure}
  \begin{center}
    \subfigure[Proto-graph $\graphT(\matrA)$.]
    {
       \epsfig{file=fg_graph_cover_construction0_1.fig.eps, scale=0.7}
    }

    \bigskip

    \subfigure[\GCCone\ based on $\graphT(\matrA)$.
               Top: $\graphT(\matrB')$.
               Bottom: $\graphT(\matrB)$.]
    {
      \epsfig{file=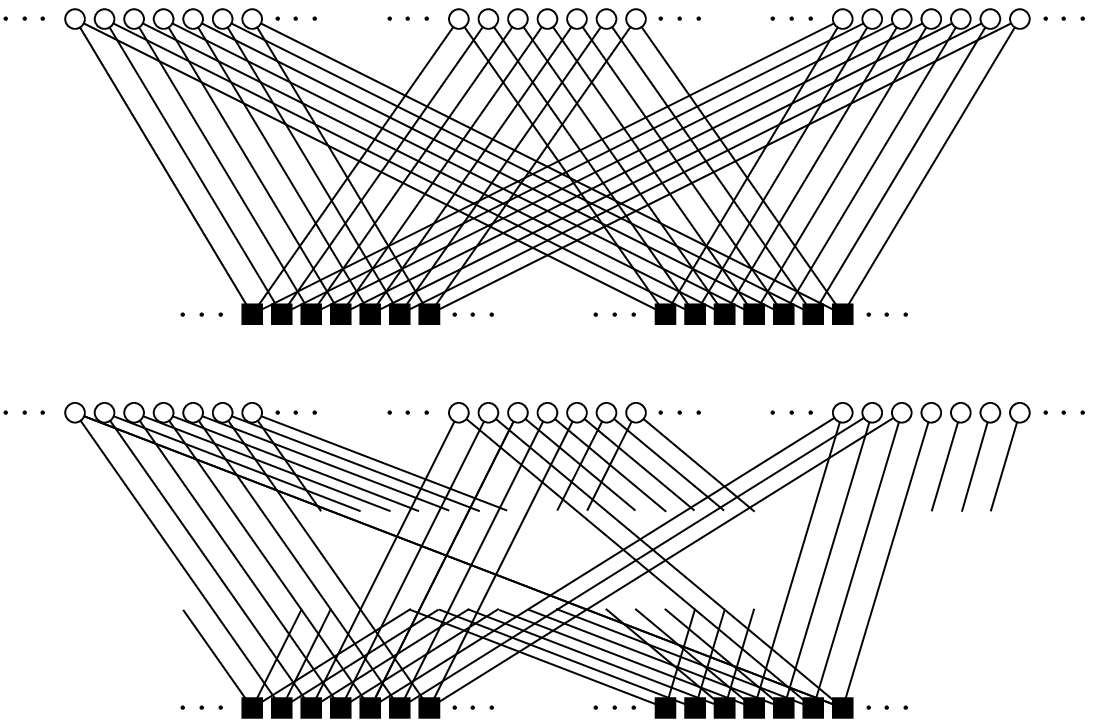, scale=0.7}
    }

    \bigskip

    \subfigure[\GCCtwo\ based on $\graphT(\matrA)$.
               Top: $\graphT(\omatrB')$.
               Bottom: $\graphT(\omatrB)$.]
    {
      \epsfig{file=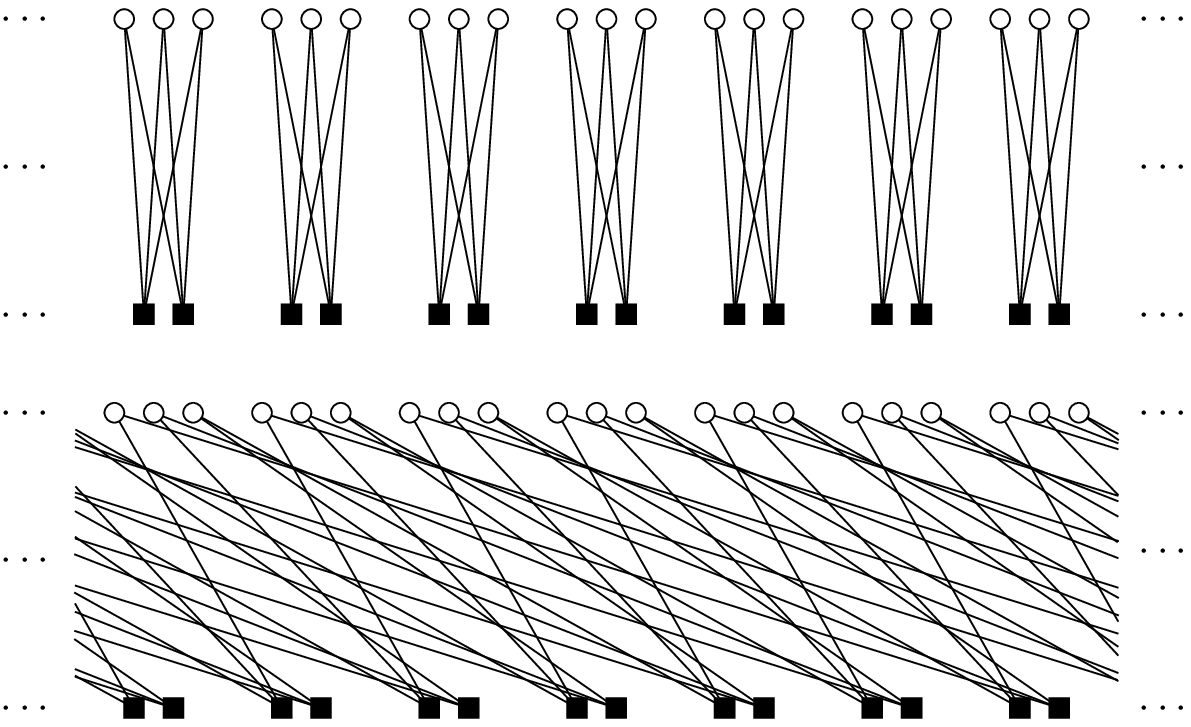, scale=0.7} 
    }
  \end{center}
  \caption{Proto-graph and graph-covers for the graph-cover constructions
    discussed in
    Example~\ref{ex:graph:cover:construction:1:infinite:1}. (Compare with the
    corresponding graphs in Figure~\ref{fig:graph:cover:constructions:1}.)}
  \label{fig:graph:cover:constructions:1:variation:2}
\end{figure}

We would like to add two comments with respect to the above example.

First, instead of defining $\matrI_s$ to be an $s$ times \emph{left}-shifted
identity matrix of size $r \times r$, we could have defined $\matrI_s$ to be
an $s$ times \emph{right}-shifted identity matrix of size $r \times
r$. Compared to the matrices and graphs described above, such a change in
definition would yield (in general) different matrices but isomorphic graphs.

Second, we note that \GCCtwo\ was termed the ``copy-and-permute'' construction
by Thorpe and his co-workers. This terminology stems from the visual
appearance of the procedure: namely, in going from
Figure~\ref{fig:graph:cover:constructions:1}(a) to
Figure~\ref{fig:graph:cover:constructions:1}(c)(top) we copy the graph several
times, and in going from Figure~\ref{fig:graph:cover:constructions:1}(c)(top)
to Figure~\ref{fig:graph:cover:constructions:1}(c)(bottom) we permute the
edges of the graph, where the permutations are done within the sets of edges
that have the same pre-image in
Figure~\ref{fig:graph:cover:constructions:1}(a).

\begin{remark}[Quasi-cyclic codes]
  \label{rem:graph:cover:construction:1:variation:1}

  Consider again the matrices that were constructed in
  Example~\ref{ex:graph:cover:construction:1}, in particular the matrix
  $\matrA$ in~\eqref{eq:ex:graph:cover:constructions:1:A:1} and its $r$-fold
  cover matrix $\matrB$
  in~\eqref{eq:ex:graph:cover:constructions:1:B:1}. Because all matrices in
  the matrix collection $\{ \matrP_{\ell} \}_{\ell \in \setL}$ are circulant,
  $\graphT(\matrB)$ represents a cyclic cover of $\graphT(\matr{A})$. Clearly,
  when seen over $\GF{2}$, the matrix $\matrHQC{r} \defeq \matrB$ is the
  parity-check matrix of a quasi-cyclic binary linear block code
  \begin{align*}
    \codeCQC{r}
      &= \Bigl\{
           \vv \in \GF{2}^{\nA \cdot r}
           \Bigm| 
           \matrHQC{r} \cdot \vv^\tr = \vect{0}^\tr
         \Bigr\}.
  \end{align*}
  Using the well-known isomorphism between the addition and multiplication of
  circulant matrices over $\GF{2}$ and the addition and multiplication of
  elements of the ring $\shortFtwoxmodr$, this code can be written
  equivalently as
  \begin{align*}
    \codeCQC{r}
      &= \Bigl\{
           \vv(X) \in \shortFtwoxmodr^{\nA}
           \Bigm| 
           \matrHQC{r}(X) \cdot \vv(X)^\tr = \vect{0}^\tr
         \Bigr\}
  \end{align*}
  with
  \begin{align*}
    \matrHQC{r}(X)
      \defeq
        \begin{bmatrix}
          X^1 & X^2 & X^4 \\
          X^6 & X^5 & X^3
        \end{bmatrix}.
  \end{align*}

  As noted above, the graphs $\graphT(\matrB)$ and $\graphT(\omatrB)$ that are
  constructed in Definition~\ref{def:graph:cover:construction:1} are
  isomorphic. Applying this observation to
  Example~\ref{ex:graph:cover:construction:1}, the matrix $\matrHbarQC{r}
  \defeq \omatrB$ with $\omatrB$
  from~\eqref{eq:ex:graph:cover:constructions:1:oB:1} is therefore the
  parity-check matrix of a binary linear block code
  \begin{align*}
    \codeCbarQC{r}
      &= \Bigl\{
           \ovv \in \GF{2}^{r \cdot \nA}
           \Bigm| 
           \matrHbarQC{r} \cdot \ovv^\tr = \vect{0}^\tr
         \Bigr\}
  \end{align*}
  that is equivalent to $\codeCQC{r}$, i.e., the codewords of $\codeCbarQC{r}$
  can be obtained from the codewords of $\codeCQC{r}$ by a suitable reordering
  of the codeword components. In terms of the matrices $\{ \matrA_{j,i}
  \}_{j,i}$, which also appear in the matrix $\omatrB$
  in~\eqref{eq:ex:graph:cover:constructions:1:oB:1}, one can verify that the
  polynomial parity-check matrix $\matrHQC{r}(X)$ can be written as
  $\matrHQC{r}(X) = \matr{0} X^0 + \matrA_{0,0} X^1 + \matrA_{0,1} X^2 +
  \matrA_{1,2} X^3 + \matrA_{0,2} X^4 + \matrA_{1,1} X^5 + \matrA_{1,0} X^6$.
  \remend
\end{remark}

\mbox{}

Besides defining finite graph covers, we can also define infinite graph
covers, as illustrated in the following examples. These infinite graph covers
will be crucial towards defining Tanner graphs of convolutional codes.

\begin{example}[Bi-infinite Toeplitz covers]
  \label{ex:graph:cover:construction:1:infinite:1}

  We continue Example~\ref{ex:graph:cover:construction:1}. However, besides
  keeping the proto-matrix $\matrA$ and the collection of matrices $\{
  \matrA_{j,i} \}_{j,i}$, we consider a different collection of matrices $\{
  \matrP_{j,i} \}_{j,i}$. Namely, we set $\matrP_{0,0} \defeq \matrT_1$,
  $\matrP_{0,1} \defeq \matrT_2$, $\matrP_{0,2} \defeq \matrT_4$,
  $\matrP_{1,0} \defeq \matrT_6$, $\matrP_{1,1} \defeq \matrT_5$,
  $\matrP_{1,2} \defeq \matrT_3$. Here $\matrT_s$ is a bi-infinite Toeplitz
  matrix with zeros everywhere except for ones in the $s$-th diagonal below
  the main diagonal, i.e., $[\matrT_s]_{j,i} = 1$ if $j = i + s$ and
  $[\matrT_s]_{j,i} = 0$ otherwise. E.g.,
  \begin{align*}
    \matrT_1
      &= \begin{bmatrix}
           \ddots & \ddots & \ddots & \ddots \\
           \ddots & \underline{0} & 0 & 0 & 0 & \\
           \ddots & 1 & \underline{0} & 0 & 0 & \ddots \\
           \ddots & 0 & 1 & \underline{0} & 0 & \ddots \\
                  & 0 & 0 & 1 & \underline{0} & \ddots \\
                  &   & \ddots & \ddots & \ddots & \ddots
         \end{bmatrix},
  \end{align*}
  where for clarity we have underlined the entries of the main diagonal.
  \begin{itemize}

  \item Using \GCCone, we obtain the matrices $\matrB' = \matrA \otimes
    \matrT_0$ and
    \begin{align}
      \matrB
        &= \begin{bmatrix}
             \matrT_1 & \matrT_2 & \matrT_4 \\
             \matrT_6 & \matrT_5 & \matrT_3
           \end{bmatrix}.
             \label{eq:ex:graph:cover:constructions:1:variation:2:B:1}
    \end{align}
    The Tanner graph $\graphT(\matrB')$, which is depicted in
    Figure~\ref{fig:graph:cover:constructions:1:variation:2}(b)(top), is
    similar to the corresponding Tanner graph in
    Figure~\ref{fig:graph:cover:constructions:1}(b)(top), but with
    bi-infinitely many independent components. Analogously, the Tanner graph
    $\graphT(\matrB)$, which is depicted in
    Figure~\ref{fig:graph:cover:constructions:1:variation:2}(b)(bottom), is
    similar to the Tanner graph shown in
    Figure~\ref{fig:graph:cover:constructions:1}(b)(bottom), but instead of
    cyclically wrapped edge connections, the edge connections are infinitely
    continued on both sides.
  
  \item Using \GCCtwo, we obtain the matrices $\omatrB' = \matrT_0 \otimes
    \matrA$ and
    \begin{align}
      \!\!\!\!\!
      \omatrB
        &= \left[\!\!
           \begin{array}{c@{\,\,\,\,\,}c@{\,}c@{\,}c@{\,}c@{\,}c@{\,}c@{\,}c@{\,}c@{\,\,\,\,\,}c}
             \ddots & \ddots & \ddots & \ddots  & 
              \ddots  & 
              \ddots & \ddots  & \ddots
             &  \\
             \ddots & \matrzero & \matrzero & \matrzero & \matrzero  & 
              \matrzero  &
              \matrzero & \matrzero & \matrzero \\
             \ddots & \matrA_{0,0} & \matrzero & \matrzero & \matrzero  & 
              \matrzero  &
              \matrzero & \matrzero & \matrzero & \ddots \\
             \ddots & \matrA_{0,1} & \matrA_{0,0} & \matrzero & \matrzero  & 
              \matrzero  &
              \matrzero & \matrzero & \matrzero & \ddots \\
             \ddots & \matrA_{1,2} & \matrA_{0,1} & \matrA_{0,0} & \matrzero 
              & 
              \matrzero  &
              \matrzero & \matrzero & \matrzero & \ddots \\
             \ddots & \matrA_{0,2} & \matrA_{1,2} & \matrA_{0,1}  & 
              \matrA_{0,0}  & 
              \matrzero  &
              \matrzero & \matrzero & \matrzero & \ddots \\
             \ddots & \matrA_{1,1} & \matrA_{0,2} & \matrA_{1,2}  & 
              \matrA_{0,1}  & 
              \matrA_{0,0}  &
              \matrzero & \matrzero & \matrzero & \ddots \\
             \ddots & \matrA_{1,0} & \matrA_{1,1} & \matrA_{0,2}  & 
              \matrA_{1,2}  &
              \matrA_{0,1}  &
              \matrA_{0,0} & \matrzero & \matrzero & \ddots \\
             &  \matrzero & \matrA_{1,0} & \matrA_{1,1} & \matrA_{0,2}  &
              \matrA_{1,2}  &
              \matrA_{0,1} & \matrA_{0,0} & \matrzero  & 
              \ddots \\
             & &  \ddots & \ddots & \ddots & \ddots  & 
              \ddots  &
              \ddots & \ddots & \ddots
           \end{array}\!\!
           \right]\!.
             \label{eq:ex:graph:cover:constructions:1:variation:2:oB:1}
    \end{align}
    The Tanner graph $\graphT(\omatrB')$, which is depicted in
    Figure~\ref{fig:graph:cover:constructions:1:variation:2}(c)(top), is
    similar to the corresponding Tanner graph in
    Figure~\ref{fig:graph:cover:constructions:1}(c)(top), but with
    bi-infinitely many independent components. Analogously, the Tanner graph
    $\graphT(\omatrB)$, which is depicted in
    Figure~\ref{fig:graph:cover:constructions:1:variation:2}(c)(bottom), is
    similar to the Tanner graph shown in
    Figure~\ref{fig:graph:cover:constructions:1}(c)(bottom), but instead of
    cyclically wrapped edge connections, the edge connections are infinitely
    continued on both sides. \exend

  \end{itemize}
\end{example}

Although it is tempting to replace in
Example~\ref{ex:graph:cover:construction:1:infinite:1} the bi-infinite
Toeplitz matrices $\matrT_s$ (whose row and column index sets equal $\Z$) by
semi-infinite Toeplitz matrices (whose row and column index sets equal $\Zp$),
note that the resulting Tanner graphs $\graphT(\matrB)$ and $\graphT(\omatrB)$
would then in general not be graph covers of $\graphT(\matrA)$. This follows
from the fact that semi-infinite Toeplitz matrices are not permutation
matrices (except for $\matrT_0$), and so some vertex degrees of
$\graphT(\matrB)$ and $\graphT(\omatrB)$ would not equal the corresponding
vertex degrees in $\graphT(\matrA)$.\footnote{As will be clear from the
  discussion later on, in this paper we take an approach where in a first step
  we construct bi-infinite Tanner graphs that are ``proper'' graph covers and
  where in a second step we obtain semi-infinite Tanner graphs by applying a
  ``shortening'' procedure to these bi-infinite Tanner graphs. Alternatively,
  one could also choose an approach based on ``improper'' graph covers. Both
  approaches have their advantages and disadvantages; we preferred to take the
  first approach.}

\begin{remark}
  \label{rem:graph:cover:connections:1}

  It turns out that the Tanner graphs in
  Figure~\ref{fig:graph:cover:constructions:1:variation:2} are infinite graph
  covers of the Tanner graphs in
  Figure~\ref{fig:graph:cover:constructions:1}. More precisely, the Tanner
  graphs $\graphT(\matrB')$, $\graphT(\matrB)$, $\graphT(\omatrB')$,
  $\graphT(\omatrB)$ in
  Figure~\ref{fig:graph:cover:constructions:1:variation:2} are graph covers of
  the corresponding Tanner graphs $\graphT(\matrB')$, $\graphT(\matrB)$,
  $\graphT(\omatrB')$, $\graphT(\omatrB)$ in
  Figure~\ref{fig:graph:cover:constructions:1}. For the Tanner graphs
  $\graphT(\matrB')$ in Figures~\ref{fig:graph:cover:constructions:1}(b)(top)
  and~\ref{fig:graph:cover:constructions:1:variation:2}(b)(top) and the Tanner
  graphs $\graphT(\omatrB')$ in
  Figures~\ref{fig:graph:cover:constructions:1}(c)(top)
  and~\ref{fig:graph:cover:constructions:1:variation:2}(c)(top), this statement
  is easily verified by inspection.

  To verify that the Tanner graph $\graphT(\omatrB)$ in
  Figure~\ref{fig:graph:cover:constructions:1:variation:2}(c)(bottom) is a
  graph cover of $\graphT(\omatrB)$ in
  Figure~\ref{fig:graph:cover:constructions:1}(c)(bottom), we apply \GCCtwo\
  with proto-matrix $\matrA$, with resulting matrix $\omatrB$, with the set
  $\setL$, with the collection of matrices $\{ \matrA_{\ell} \}_{\ell \in
    \setL}$, and with the collection of permutation matrices $\{ \matrP_{\ell}
  \}_{\ell \in \setL}$ as follows. Namely, we let the proto-matrix $\matrA$ be
  the matrix from~\eqref{eq:ex:graph:cover:constructions:1:oB:1} (there
  denoted by $\omatrB$), we let the resulting matrix $\omatrB$ be the matrix
  in~\eqref{eq:ex:graph:cover:constructions:1:variation:2:oB:1} (there denoted
  by $\omatrB$), we define $\setL \defeq \{ 0, 1 \}$, we select
  \begin{align}
    \matrA_0
      &= \left[
         \begin{array}{@{\;}c@{\;\;}c@{\;\;}c@{\;\;}c@{\;\;}c@{\;\;}c@{\;\;}c@{\;}}
           \matr{0} & \matr{0} & \matr{0} & \matr{0} &
           \matr{0} & \matr{0} & \matr{0} \\
           \matrA_{0,0} & \matr{0} & \matr{0} & \matr{0} & 
           \matr{0} & \matr{0} & \matr{0} \\
           \matrA_{0,1} & \matrA_{0,0} & \matr{0} & \matr{0} &
           \matr{0} & \matr{0} & \matr{0} \\
           \matrA_{1,2} & \matrA_{0,1} & \matrA_{0,0} & \matr{0} &
           \matr{0} & \matr{0} & \matr{0} \\
           \matrA_{0,2} & \matrA_{1,2} & \matrA_{0,1} & \matrA_{0,0} &
           \matr{0} & \matr{0} & \matr{0} \\
           \matrA_{1,1} & \matrA_{0,2} & \matrA_{1,2} & \matrA_{0,1} &
           \matrA_{0,0} & \matr{0} & \matr{0} \\
           \matrA_{1,0} & \matrA_{1,1} & \matrA_{0,2} & \matrA_{1,2} &
           \matrA_{0,1} & \matrA_{0,0} & \matr{0}
         \end{array}
         \right],
           \label{eq:rem:graph:cover:connections:1:A:0} \\[0.25cm]
    \matrA_1
      &= \left[
         \begin{array}{@{\;}c@{\;\;}c@{\;\;}c@{\;\;}c@{\;\;}c@{\;\;}c@{\;\;}c@{\;}}
           \matr{0} & \matrA_{1,0} & \matrA_{1,1} & \matrA_{0,2} &
           \matrA_{1,2} & \matrA_{0,1} & \matrA_{0,0} \\
           \matr{0} & \matr{0} & \matrA_{1,0} & \matrA_{1,1} &
           \matrA_{0,2} & \matrA_{1,2} & \matrA_{0,1} \\
           \matr{0} & \matr{0} & \matr{0} & \matrA_{1,0} &
           \matrA_{1,1} & \matrA_{0,2} & \matrA_{1,2} \\
           \matr{0} & \matr{0} & \matr{0} & \matr{0} &
           \matrA_{1,0} & \matrA_{1,1} & \matrA_{0,2} \\
           \matr{0} & \matr{0} & \matr{0} & \matr{0} &
           \matr{0} & \matrA_{1,0} & \matrA_{1,1} \\
           \matr{0} & \matr{0} & \matr{0} & \matr{0} &
           \matr{0} & \matr{0} & \matrA_{1,0} \\
           \matr{0} & \matr{0} & \matr{0} & \matr{0} &
           \matr{0} & \matr{0} & \matr{0}
         \end{array}
         \right],
           \label{eq:rem:graph:cover:connections:1:A:1}
  \end{align}
  and we select $\matrP_0 = \matrT_0$ and $\matrP_1 = \matrT_1$, where
  $\matrT_s$ was defined in
  Example~\ref{ex:graph:cover:construction:1:infinite:1}. Clearly, $\matrA =
  \matrA_0 + \matrA_1 \ \text{(in $\Z$)}$.\footnote{Note that a non-zero block
    diagonal of $\matrA$ would be put in $\matrA_0$.} With this we have
  \begin{align*}
    \omatrB
      &= \matrP_0 \otimes \matrA_0
         +
         \matrP_1 \otimes \matrA_1
       = \left[\!\!
         \begin{array}{c@{\;\;\;\;}c@{\;\;}c@{\;\;}c@{\;\;}c@{\;\;}c}
           \ddots & \ddots & \ddots & \ddots  & 
            &  \\
           \ddots & \matrA_0 & \matr{0} & \matr{0}  & 
            \matr{0}  & \\
           \ddots & \matrA_1 & \matrA_0 & \matr{0}  & 
            \matr{0} &
            \ddots \\
           \ddots & \matr{0} & \matrA_1 & \matrA_0 
            & 
            \matr{0}  &
            \ddots \\
            & \matr{0} & \matr{0}  & 
            \matrA_1  & 
            \matrA_0  &
            \ddots \\
           & &  \ddots & \ddots & 
            \ddots  &
            \ddots 
         \end{array}\!\!
         \right] \!\! ,
  \end{align*}
  and one can verify that this matrix equals the matrix
  in~\eqref{eq:ex:graph:cover:constructions:1:variation:2:oB:1} (there denoted
  by $\omatrB$), which means that $\graphT(\omatrB)$ is indeed an infinite
  cover of $\graphT(\matrA)$. We remark that, interestingly, in this process
  we have shown how a certain \GCCtwo\ graph cover of a proto-matrix can be
  written as a \GCCtwo\ graph cover of a certain \GCCtwo\ graph cover of that
  proto-matrix.

  Finally, a similar argument shows that the Tanner graph $\graphT(\matrB)$ in
  Figure~\ref{fig:graph:cover:constructions:1:variation:2}(b)(top) is a
  graph cover of the Tanner graph in
  Figure~\ref{fig:graph:cover:constructions:1}(b)(top), also denoted by
  $\graphT(\matrB)$. \remend
\end{remark}

\begin{figure*}
  \begin{center}
    \begin{minipage}[c]{\linewidth}
      \subfigure[First decomposition of the matrix $\matrA$
                 into the matrices $\matrA_0$ and $\matrA_1$.]
      {
        \begin{minipage}[c]{0.48\linewidth}
          \begin{center}
            \epsfig{file=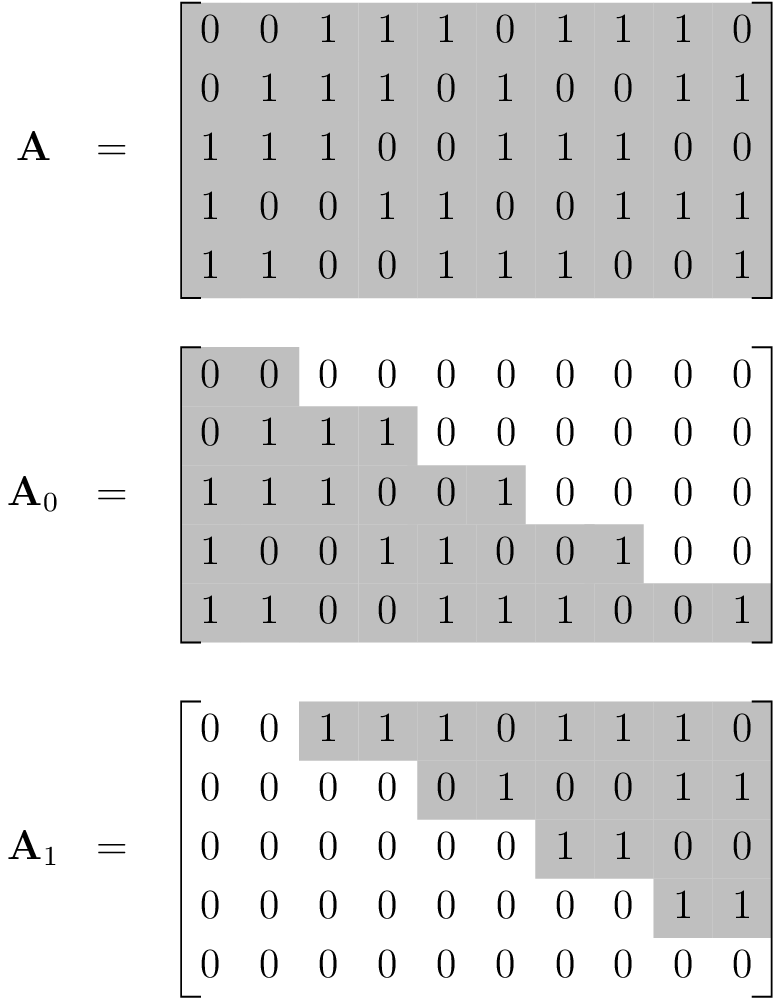,
                    width=0.60\linewidth}

            \mbox{}
          \end{center}
        \end{minipage}
      }
      \begin{minipage}[c]{0.02\linewidth}
      \end{minipage}
      \subfigure[Second decomposition of the matrix $\matrA$
                 into the matrices $\matrA_0$ and $\matrA_1$.]
      {
        \begin{minipage}[c]{0.48\linewidth}
          \begin{center}
            \epsfig{file=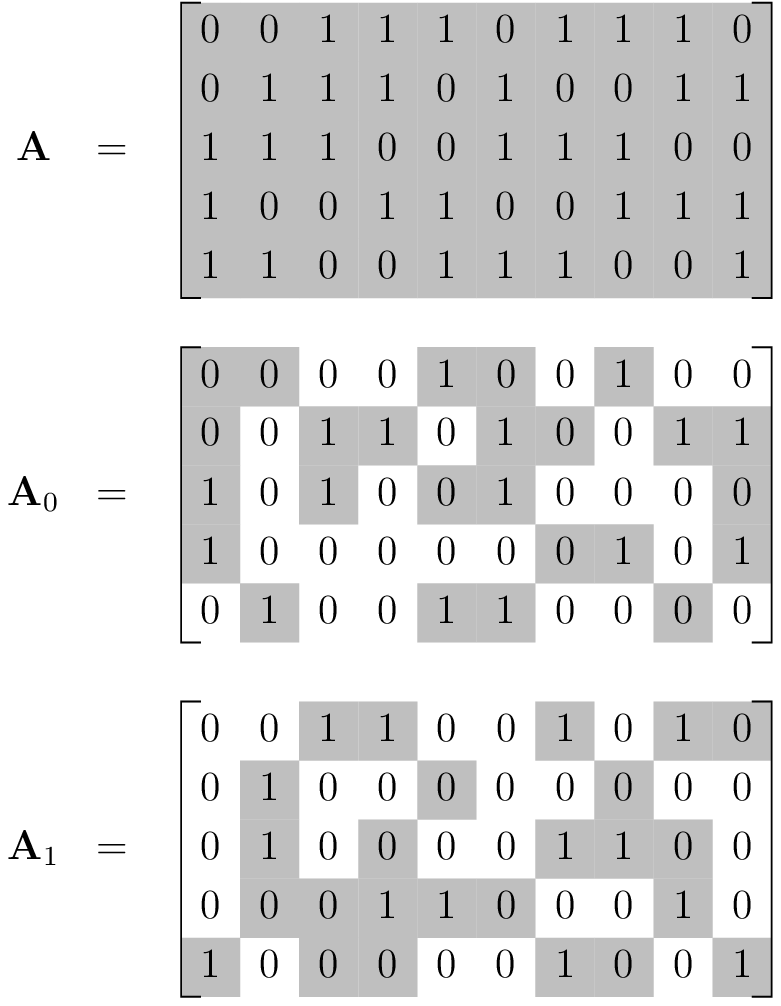,
                    width=0.60\linewidth}

            \mbox{}
          \end{center}
        \end{minipage}
      }

      \subfigure[Part of the matrix $\omatrB$ based on the 
                 first decomposition of $\matrA$.]
      {
        \begin{minipage}[c]{0.48\linewidth}
          \begin{center}
            \epsfig{file=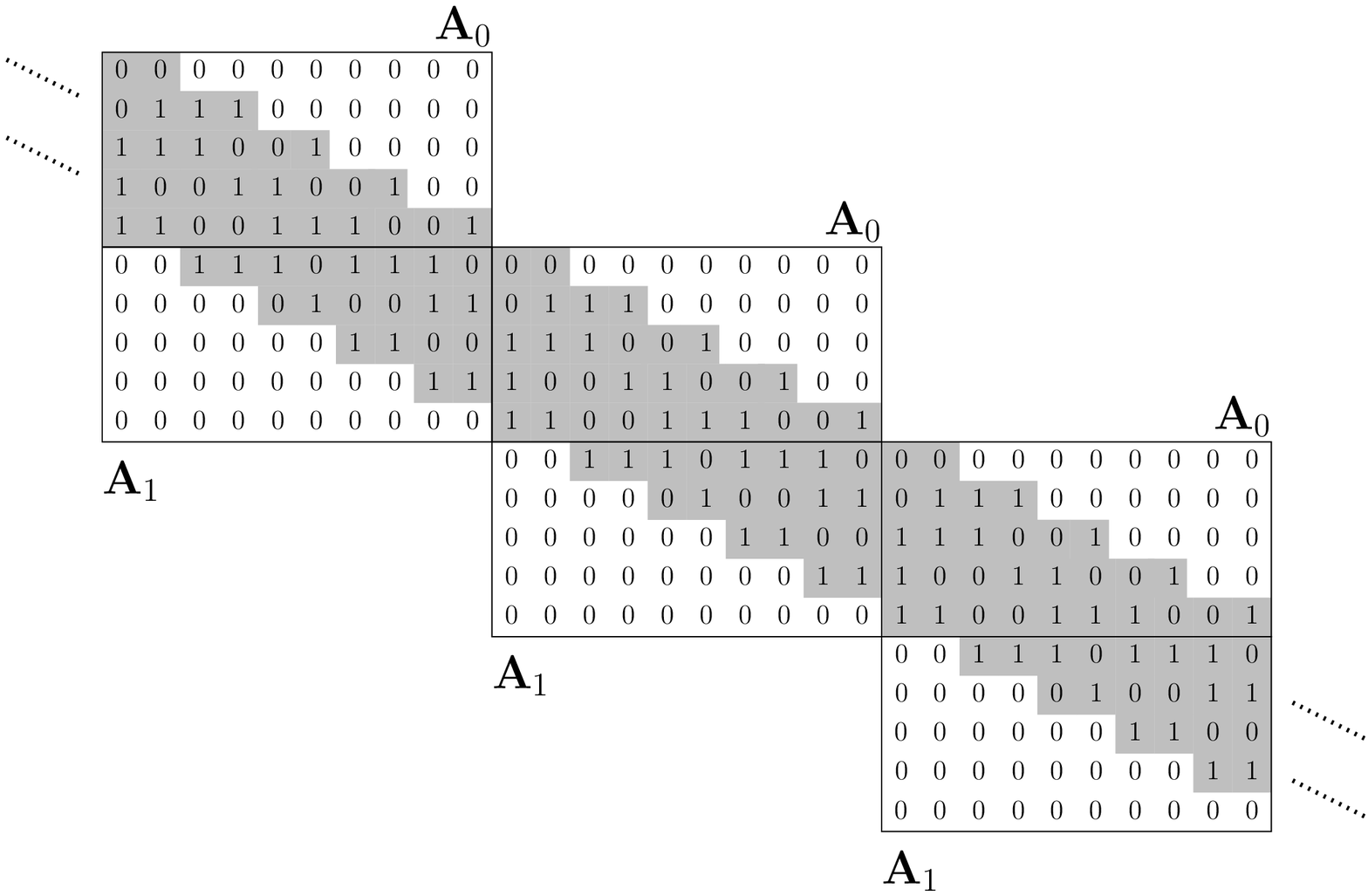,
                    width=\linewidth}

            \mbox{}
          \end{center}
        \end{minipage}
      }
      \begin{minipage}[c]{0.02\linewidth}
      \end{minipage}
      \subfigure[Part of the matrix $\omatrB$ based on the 
                 second decomposition of $\matrA$.]
      {
        \begin{minipage}[c]{0.48\linewidth}
          \begin{center}
            \epsfig{file=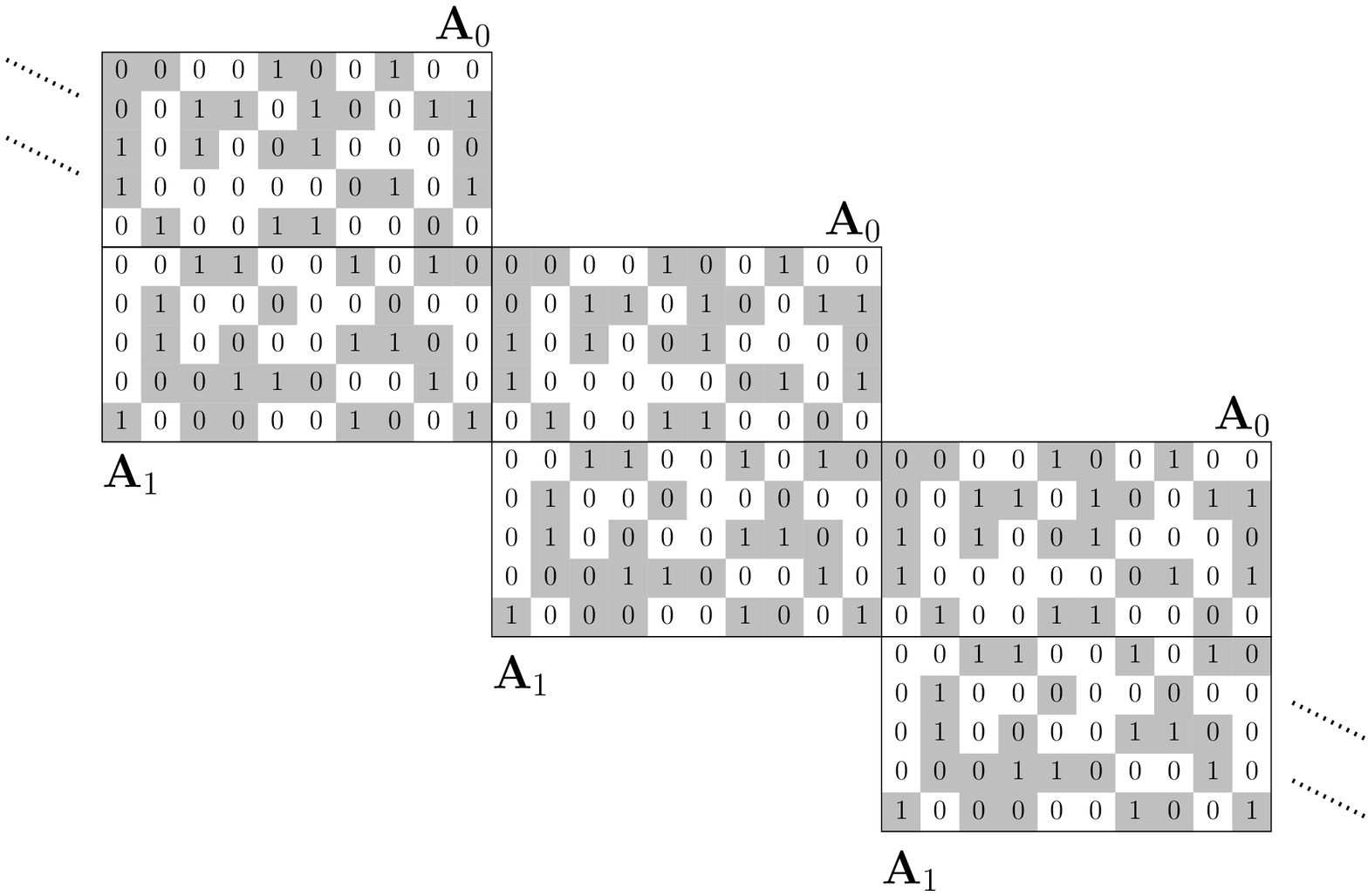,
                    width=\linewidth}

            \mbox{}
          \end{center}
        \end{minipage}
      }
    \end{minipage}
  \end{center}
  \caption{Matrices appearing in
    Example~\ref{ex:graph:cover:construction:2}. (See main text for details.)}
  \label{fig:graph:cover:constructions:2}
\end{figure*}

There are many other ways of writing a proto-matrix $\matrA$ as a sum of a
collection of matrices $\{ \matrA_{\ell} \}_{\ell \in \setL}$. The next
example discusses two such possibilities.

\begin{example}
  \label{ex:graph:cover:construction:2}

  Consider the proto-matrix
  \begin{align*}
    \matr{A}
      \defeq
        \begin{bmatrix}
          0 & 0 & 1 & 1 & 1 & 0 & 1 & 1 & 1 & 0 \\
          0 & 1 & 1 & 1 & 0 & 1 & 0 & 0 & 1 & 1 \\
          1 & 1 & 1 & 0 & 0 & 1 & 1 & 1 & 0 & 0 \\
          1 & 0 & 0 & 1 & 1 & 0 & 0 & 1 & 1 & 1 \\
          1 & 1 & 0 & 0 & 1 & 1 & 1 & 0 & 0 & 1
        \end{bmatrix}
  \end{align*}
  that is shown in Figure~\ref{fig:graph:cover:constructions:2}(a), and that
  also appeared in Example~\ref{example:unwrap:1}. Its Tanner graph
  $\graphT(\matrA)$ is $(3,6)$-regular, i.e., all variable nodes have degree
  $3$ and all check nodes have degree $6$. Let $\setL = \{ 0, 1 \}$, and
  consider the collection of matrices $\{ \matrP_{\ell} \}_{\ell \in \setL}$
  with $\matrP_0 = \matrT_0$ and $\matrP_1 = \matrT_1$, where the matrices
  $\matrT_0$ and $\matrT_1$ are defined as in
  Example~\ref{ex:graph:cover:construction:1:infinite:1}. In the following, we
  look at two different choices of the collection of matrices $\{
  \matrA_{\ell} \}_{\ell \in \setL}$.
  \begin{itemize}
 
  \item Figure~\ref{fig:graph:cover:constructions:2}(c) shows a typical part
    of the matrix $\omatrB$ that is obtained when \GCCtwo\ is used to construct
    a graph cover of $\matrA$ with the collection of matrices $\{
    \matrA_{\ell} \}_{\ell \in \setL}$ defined as shown in
    Figure~\ref{fig:graph:cover:constructions:2}(a).

  \item Figure~\ref{fig:graph:cover:constructions:2}(d) shows a typical part
    of the matrix $\omatrB$ when \GCCtwo\ is used to construct a graph cover
    of $\matrA$ with the collection of matrices $\{ \matrA_{\ell} \}_{\ell \in
      \setL}$ defined as shown in
    Figure~\ref{fig:graph:cover:constructions:2}(b).

  \end{itemize}
  Overall, because of the choice of the collection $\{ \matrP_{\ell} \}_{\ell
    \in \setL}$, the support of both matrices $\omatrB$ possesses a banded
  diagonal structure. Moreover, the different choices of the collection $\{
  \matrA_{\ell} \}_{\ell \in \setL}$ leads to a somewhat narrower banded
  diagonal structure in the first case compared to the second case. \exend
\end{example}

The next example makes a crucial observation; namely, it shows that the above
graph-cover constructions can be applied repeatedly to obtain additional
interesting classes of Tanner graphs.

\begin{example}[Iterated Graph-Cover Construction]
  \label{ex:graph:cover:construction:3}

  Starting with the proto-matrix $\matr{A}$ from
  Example~\ref{ex:graph:cover:construction:1}, we consider two iterated
  graph-cover constructions. In the first case, we apply \GCCone\ and then
  \GCCtwo, and in the second case we apply \GCCtwo\ twice.
  \begin{itemize}

  \item Consider the matrix $\matrB$ obtained from the matrix $\matrA$ using
    \GCCone, like in Example~\ref{ex:graph:cover:construction:1}. The
    resulting matrix $\matrB$ is shown
    in~\eqref{eq:ex:graph:cover:constructions:1:B:1} and will be called
    $\matrA^{(1)}$ in this example, since it is considered to be a
    proto-matrix by itself,
    cf.~Figure~\ref{fig:graph:cover:construction:3}(a). Based on the ``cutting
    line'' shown in Figure~\ref{fig:graph:cover:construction:3}(a), we define
    the matrices $\matrA^{(1)}_0$ and $\matrA^{(1)}_1$ as follows: the
    non-zero part of $\matrA^{(1)}_0$ equals the non-zero part of the lower
    triangular part of $\matrA^{(1)}$ and the non-zero part of
    $\matrA^{(1)}_1$ equals the non-zero part of the upper triangular part of
    $\matrA^{(1)}$. (Clearly, $\matrA^{(1)} = \matr{A}^{(1)}_0 +
    \matr{A}^{(1)}_1$.) Applying the procedure from
    Example~\ref{ex:graph:cover:construction:2},
    Figure~\ref{fig:graph:cover:construction:3}(c) shows a typical part of the
    matrix $\omatrB^{(1)}$ that is obtained when \GCCtwo\ is used to construct
    a graph cover of $\matrA^{(1)}$.

  \item Consider the graph-cover $\omatrB$ obtained from $\matrA$ using
    \GCCtwo, like in Example~\ref{ex:graph:cover:construction:1}. The
    resulting matrix $\omatrB$ is shown
    in~\eqref{eq:ex:graph:cover:constructions:1:oB:1} and will be called
    $\matrA^{(2)}$ in this example, since it is considered to be a
    proto-matrix by itself,
    cf.~Figure~\ref{fig:graph:cover:construction:3}(b). Based on the ``cutting
    line'' shown in Figure~\ref{fig:graph:cover:construction:3}(b), we define
    the matrices $\matrA^{(2)}_0$ and $\matrA^{(2)}_1$ as follows: the
    non-zero part of $\matrA^{(2)}_0$ equals the non-zero part of the lower
    triangular part of $\matrA^{(2)}$ and the non-zero part of
    $\matrA^{(2)}_1$ equals the non-zero part of the upper triangular part of
    $\matrA^{(2)}$. (Clearly, $\matrA^{(2)} = \matr{A}^{(2)}_0 +
    \matr{A}^{(2)}_1$.)  Applying the procedure from
    Example~\ref{ex:graph:cover:construction:2},
    Figure~\ref{fig:graph:cover:construction:3}(d) shows a typical part of the
    matrix $\omatrB^{(2)}$ that is obtained when \GCCtwo\ is used to construct
    a graph cover of $\matrA^{(2)}$.

  \end{itemize}

  We observe a large difference in the positions of the non-zero entries in
  $\omatrB^{(1)}$ and $\omatrB^{(2)}$.
  \begin{itemize}

  \item In the first case, the two graph-cover constructions are
    ``incompatible'' and the positions of the non-zero entries in
    $\omatrB^{(1)}$ follow a ``non-simple'' or ``pseudo-random'' pattern. As
    we will see in Example~\ref{ex:iterated:gcc:LDPC:code:construction:1}
    with the help of simulation results, such Tanner graphs can lead to
    time-varying LDPC convolutional codes with very good performance.

  \item In the second case, the two graph-cover constructions are
    ``compatible'' in the sense that $\omatrB^{(2)}$ can be obtained from the
    proto-matrix $\matrA$ by applying \GCCtwo\ with suitable matrix
    collections $\{ \matrA_{\ell} \}_{\ell \in \setL}$ and $\{ \matrP_{\ell}
    \}_{\ell \in \setL}$. As such, the positions of the non-zero entries of
    $\omatrB^{(2)}$ follow a relatively ``simple'' or ``non-random'' pattern,
    which leads to a time-invariant LDPC convolutional code. \exend

  \end{itemize}

\end{example}

\begin{figure}
  \begin{center}
    \subfigure[Matrix $\matrA^{(1)}$]
    {
      \includegraphics[width=0.46\columnwidth,]
                      {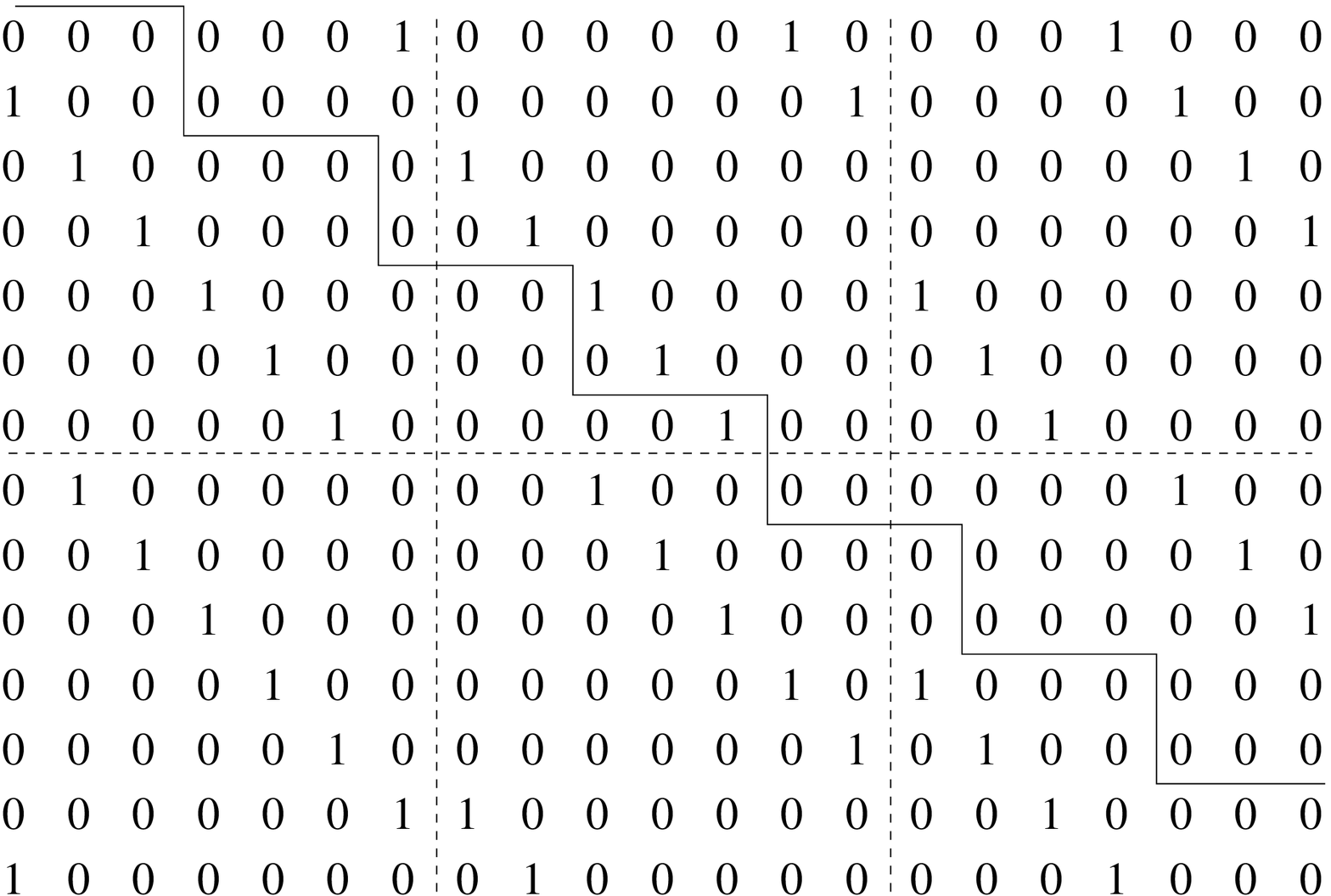}
 
    }
    \subfigure[Matrix $\matrA^{(2)}$]
    {
      \includegraphics[width=0.46\columnwidth]
                      {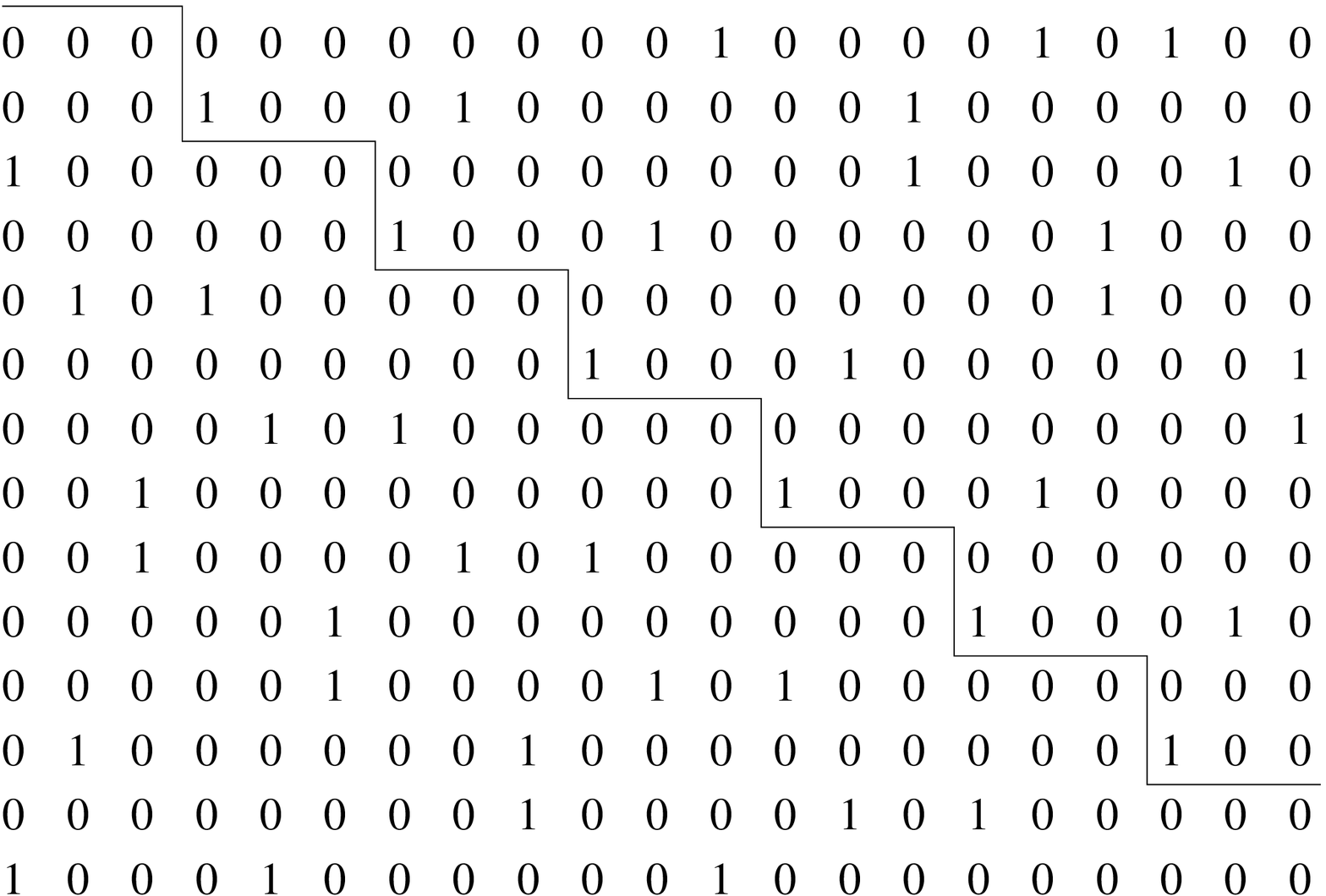}
    }

    \subfigure[Part of matrix $\omatrB^{(1)}$]
    {
      \includegraphics[width=0.46\columnwidth,]
                      {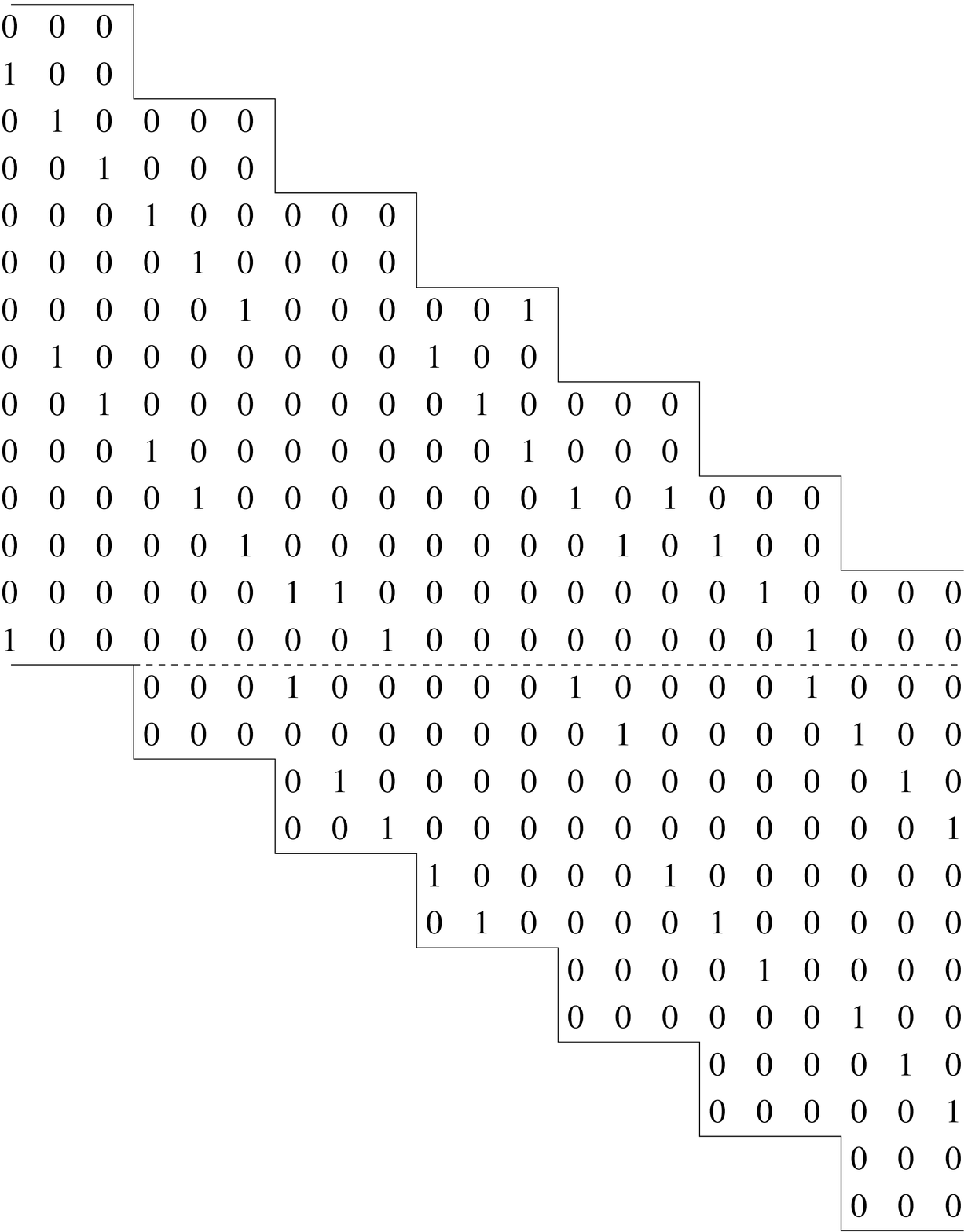} 
    }
    \subfigure[Part of matrix $\omatrB^{(2)}$]
    {
      \includegraphics[width=0.46\columnwidth,]
                      {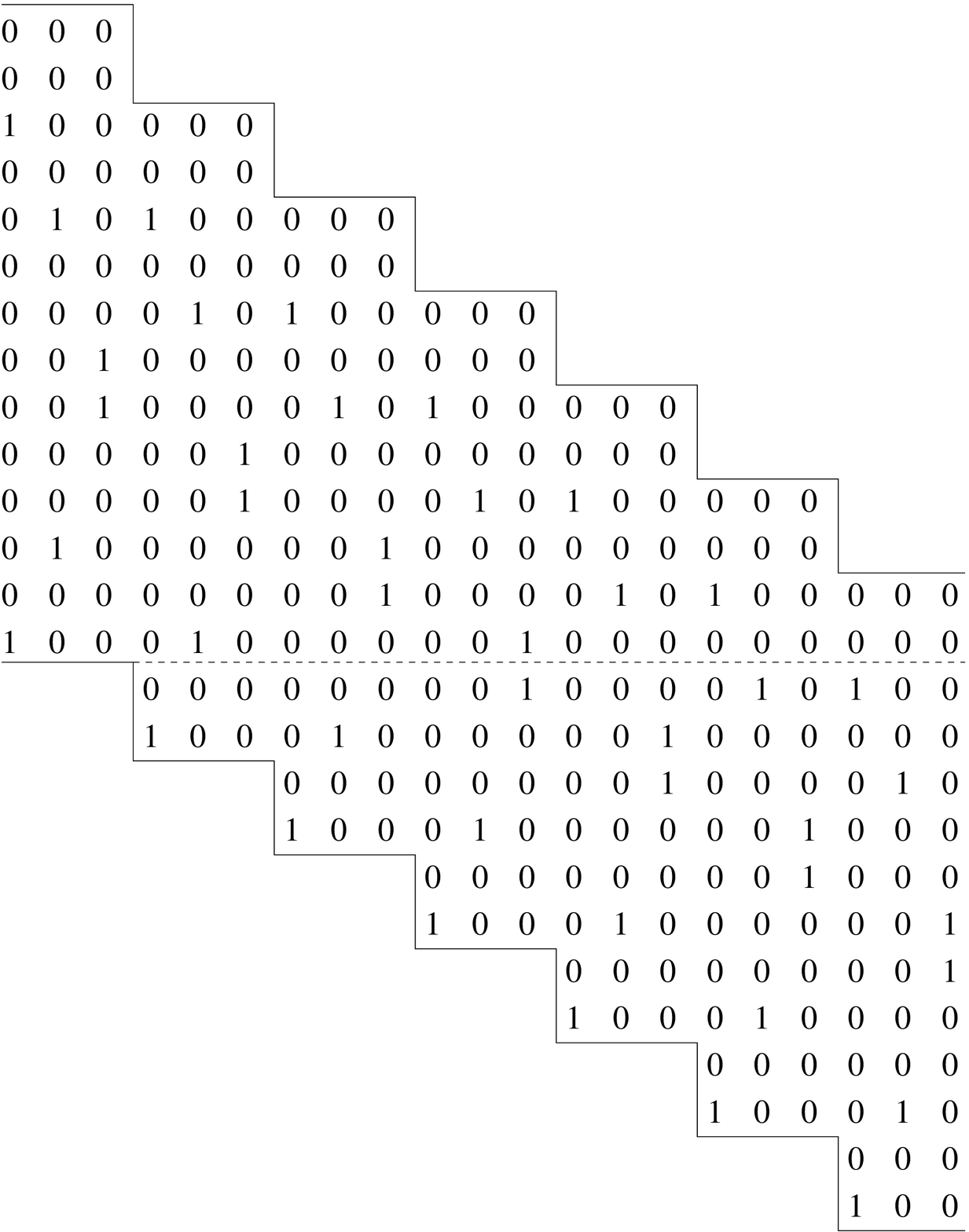}
    }
  \end{center}
  \caption{Matrices appearing in
    Example~\ref{ex:graph:cover:construction:3}. (See main text for details.)}
  \label{fig:graph:cover:construction:3}
\end{figure}

The above procedure of obtaining two matrices that add up to a matrix is
called ``cutting a matrix''. Actually, we will also use this term if there is
no simple cutting line, as in the above examples, and also if the matrix is
written as the sum of more than two matrices
(cf.~Example~\ref{example:unwrap:1} and the paragraphs after it).

\subsection{Revisiting the Tanner and the JFZ Unwrapping Techniques}
\label{sec:unwrapping:1:rev}

In Section~\ref{sec:unwrapping:1} we introduced two techniques, termed the
Tanner and the JFZ unwrapping techniques, to derive convolutional codes from
block codes. In this subsection we revisit these unwrapping techniques. In
particular, we show how they can be cast in terms of graph covers and how the
two unwrapping techniques are connected.

Because of the importance of the coding-theoretic notion of
shortening~\cite{MacWilliams:Sloane:77} for this subsection, we briefly
revisit this concept. Let $\matrH$ be a parity-check matrix that defines some
length-$n$ binary code $\code{C}$. We say that the length-$(n\!-\!1)$ code
$\code{C}'$ is obtained by shortening $\code{C}$ at position $i$ if
\begin{align*}
  \code{C}'\!\!
    &= \Bigm\{\!\!
         (v_0, \ldots, v_{i-1}, v_{i+1}, \ldots, v_{n-1}) \in \GF{2}^{n-1}
       \!\!\Bigm|\!
         \vv \in \code{C}, v_i = 0
       \!\!\Bigm\}\!.
\end{align*}
In terms of parity-check matrices, a possible parity-check matrix $\matrH'$ of
$\code{C}'$ is obtained by deleting the $i$-th column of $\matrH$. In terms of
Tanner graphs, this means that the Tanner graph $\graphT(\matrH')$ is obtained
from $\graphT(\matrH)$ by removing the $i$-th variable node, along with its
incident edges. In the following, we will also use the term ``shortening'' to
denote this graph modification procedure.

Now, to explain the Tanner unwrapping technique in terms of graph covers,
consider the quasi-cyclic block code $\codeCQC{r}$ defined by the
polynomial parity-check matrix $\matrHQC{r}(X)$ of size $\mA \times \nA$,
i.e.,
\begin{align*}
  \codeCQC{r}
    &= \Bigl\{
         \vv(X) \in \shortFtwoxmodr^{\nA}
         \Bigm| 
         \matrHQC{r}(X) \cdot \vv(X)^\tr = \vect{0}^\tr
       \Bigr\},
\end{align*}
where the polynomial operations are performed modulo $X^r - 1$ (see also
Remark~\ref{rem:graph:cover:construction:1:variation:1}). As already mentioned
in Section~\ref{sec:unwrapping:1}, the Tanner unwrapping technique is simply
based on dropping these modulo computations. More precisely, with a
quasi-cyclic block code $\codeCQC{r}$, we associate the convolutional code
\begin{align*}
  \codeCconv
    &= \Bigl\{
         \vv(D) \in \GF{2}[D]^{\nA}
         \Bigm|
         \matrHconv(D) \cdot \vv(D)^\tr = \vect{0}^\tr
       \Bigr\}
\end{align*}
with polynomial parity-check matrix
\begin{align*}
  \matrHconv(D)
    &\defeq
       \left.
         \matrHQC{r}(X)
       \right|_{X = D}.
\end{align*}
Again, the change of indeterminate from $X$ to $D$ indicates the lack of
modulo $D^r-1$ operations.

In the following we will give, with the help of an example, two
interpretations of the Tanner unwrapping technique in terms of graph covers.

\begin{example}
  \label{ex:Tanner:unwrapping:1:rev:1}

  Unwrapping the quasi-cyclic block code $\codeCQC{r}$ that was considered in
  Remark~\ref{rem:graph:cover:construction:1:variation:1}, we obtain a
  rate-$1/3$ time-invariant convolutional code
  \begin{align*}
    \codeCconv
      &= \Bigl\{
           \vv(D) \in \GF{2}[D]^{\nA}
           \Bigm| 
           \matrHconv(D) \cdot \vv(D)^\tr = \vect{0}^\tr
         \Bigr\}
  \end{align*}
  with polynomial parity-check matrix
  \begin{align*}
    \matrHconv(D)
      \defeq
        \begin{bmatrix}
          D^1 & D^2 & D^4 \\
          D^6 & D^5 & D^3
        \end{bmatrix}.
  \end{align*}
  Consider now the infinite graph covers that were constructed in
  Example~\ref{ex:graph:cover:construction:1:infinite:1} using \GCCone, in
  particular $\graphT(\matrB)$. Let $\code{C}\big( \graphT(\matrB) \big)$ be
  the set of codewords defined by the Tanner graph $\graphT(\matrB)$. Then the
  convolutional code $\codeCconv$ is a shortened version of $\code{C}\big(
  \graphT(\matrB) \big)$ where all codeword bits corresponding to negative
  time indices have been shortened. Therefore, the Tanner graph of
  $\codeCconv$ is given by the Tanner graph in
  Figure~\ref{fig:graph:cover:constructions:1:variation:2}(b)(bottom), where
  all bit nodes with negative time indices, along with their incident edges,
  are removed. Clearly, this bit-node and edge removal process implies
  decreasing the degrees of some check nodes. In fact, some check nodes become
  obsolete, because their degree is decreased to zero. \exend
\end{example}

Therefore, one interpretation of the Tanner unwrapping technique in terms of
graph covers is that the Tanner graph of the convolutional code is obtained by
taking a suitable graph cover of the same proto-graph that was used to
construct the quasi-cyclic LDPC code, along with some suitable shortening.

\begin{example}
  \label{ex:Tanner:unwrapping:1:rev:2}

  We continue Remark~\ref{rem:graph:cover:construction:1:variation:1} and
  Example~\ref{ex:Tanner:unwrapping:1:rev:1}. Clearly, in the same way as the
  block code $\codeCbarQC{r}$ is equivalent to the block code $\codeCQC{r}$,
  we can define a code $\codeCbarconv$ (with parity-check matrix
  $\matrHbarconv$) that is equivalent to $\codeCconv$. The observations in
  Remark~\ref{rem:graph:cover:connections:1} and
  Example~\ref{ex:Tanner:unwrapping:1:rev:1} can then be used to show that the
  Tanner graph of $\matrHbarconv$ equals a graph cover of the Tanner graph
  $\matrHbarQC{r}$, along with some suitable shortening. \exend
\end{example}

Therefore, the second interpretation of the Tanner unwrapping in terms of graph
covers is that the Tanner graph of the convolutional code is obtained by
taking a suitable graph cover of the Tanner graph of the quasi-cyclic code,
along with some suitable shortening.

Now turning our attention to the JFZ unwrapping technique, recall from
Section~\ref{sec:unwrapping:1} that this method is based on writing a
parity-check matrix $\matrHbar$ of some block code $\codeCbar$ as the sum
$\matrHbar = \sum_{\ell \in \setL} \matrH_{\ell} \ (\text{in $\Z$})$ of a
collection of matrices $\{ \matrH_{\ell} \}_{\ell \in \setL}$. The
convolutional code is then defined to be
\begin{align}
  \codeCbarconv
    &\defeq
       \Bigl\{
         \ovv \in \GF{2}^{\infty}
         \Bigm| 
         \matrHbarconv \cdot \ovv^\tr = \vect{0}^\tr
       \Bigr\},
         \label{eq:code:C:conv:JFZ:unwrapping:1:rev}
\end{align}
where
\begin{align} 
  \matrHbarconv
    &\defeq
      \left[
      \begin{array}{@{}c@{\;}c@{\;\;}c@{\;\;}c@{\;\;}c@{\;\;}c@{}}
         \matr{H}_0 \hfill & & & & & \\
         \matr{H}_1 \hfill & \matr{H}_0 \hfill & & & & \\
         \vdots & \vdots & \ddots & & & \\
         \matr{H}_{|\setL|-1} \hfill & \matr{H}_{|\setL|-2} \hfill & 
         \ldots & \matr{H}_0 \hfill & & \\
         & \matr{H}_{|\setL|-1} \hfill & \matr{H}_{|\setL|-2} \hfill & 
         \ldots & \matr{H}_{0} \hfill & \\
         & & \ddots & \ddots & \ddots & \ddots 
      \end{array}
      \right].
         \label{eq:code:C:conv:JFZ:unwrapping:1:rev:pcm:1}
\end{align}

With the help of an example, we now explain how the JFZ unwrapping technique
can be cast in terms of graph-covers.

\begin{example}
  \label{ex:JFZ:unwrapping:2:rev}

  Consider the infinite graph covers that were constructed using \GCCtwo\ in
  Example~\ref{ex:graph:cover:construction:1:infinite:1}, in particular
  $\graphT(\omatrB)$. Let $\code{C}\big( \graphT(\omatrB) \big)$ be the set of
  valid assignments to the Tanner graph $\graphT(\omatrB)$. Moreover, let
  $\matrHbar \defeq \matrH_0 + \matrH_1 + \cdots + \matrH_6 \defeq \matr{0} +
  \matrA_{0,0} + \matrA_{0,1} + \matrA_{1,2} + \matrA_{0,2} + \matrA_{1,1} +
  \matrA_{1,0}$, and let $\codeCbarconv$ be defined as
  in~\eqref{eq:code:C:conv:JFZ:unwrapping:1:rev}. Then the code
  $\codeCbarconv$ is a shortened version of $\code{C}\big( \graphT(\omatrB)
  \big)$, where all codeword bits corresponding to negative time indices have
  been shortened. Therefore, the Tanner graph of $\codeCbarconv$ is given by
  the Tanner graph in
  Figure~\ref{fig:graph:cover:constructions:1:variation:2}(c)(bottom), where
  all the bit nodes with negative time indices are shortened. \exend
\end{example}

In order to connect the unwrapping techniques due to Tanner and due to JFZ, we
show now, with the help of an example, that in fact the unwrapping technique
due to Tanner can be seen as a \emph{special case} of the unwrapping technique
due to JFZ.\footnote{We leave it as an exercise for the reader to show the
  validity of this connection beyond this specific example.}

\begin{example}
  Consider the quasi-cyclic block code defined by the parity-check matrix
  $\matrHbarQC{r} \defeq \omatrB$, where $\omatrB$ was defined
  in~\eqref{eq:ex:graph:cover:constructions:1:oB:1}. Applying the JFZ
  unwrapping technique with the matrix decomposition $\matrHbarQC{r} =
  \matrA_0 + \matrA_1 \ \text{(in $\Z$)}$, with $\matrA_0$ defined
  in~\eqref{eq:rem:graph:cover:connections:1:A:0} and $\matrA_1$ defined
  in~\eqref{eq:rem:graph:cover:connections:1:A:1}, $\matrHbarconv$ turns out
  to equal a submatrix of $\omatrB$
  in~\eqref{eq:ex:graph:cover:constructions:1:variation:2:oB:1}, namely the
  submatrix of $\omatrB$ where the row and column index set are equal to
  $\Zp$. However, the code defined by $\matrHbarconv$ is equivalent to the
  code defined by the Tanner unwrapping technique applied to the quasi-cyclic
  code defined by $\matrHbarQC{r}$. \exend
\end{example}

Therefore, the unwrapping technique due to JFZ is more general. In fact,
whereas the Tanner unwrapping technique leads to \emph{time-invariant}
convolutional codes, the unwrapping technique due to JFZ can, depending on the
parity-check matrix decomposition and the internal structure of the terms in
the decomposition, lead to \emph{time-varying} convolutional codes with
non-trivial period.\footnote{Of course, if the underlying quasi-cyclic block
  code is suitably chosen, then also the Tanner unwrapping technique can yield a
  time-varying convolutional code; however, we do not consider this option
  here.}

Despite the fact that the unwrapping technique due to Tanner is a special case
of the unwrapping technique due to JFZ, it is nevertheless helpful to have
both unwrapping techniques at hand, because sometimes one framework can be
more convenient than the other. We will use both perspectives in the next
section.

We conclude this section with the following remarks. 
\begin{itemize}

\item Although most of the examples in this section have regular bit node
  degree~$2$ and regular check node degree~$3$, there is nothing special about
  this choice of bit and check node degrees; any other choice would work
  equally well.

\item Although all polynomial parity-check matrices that appear in this
  section contain only monomials, this is not required, i.e., the developments
  in this section work equally well for polynomial parity-check matrices
  containing the zero polynomial, monomials, binomials, trinomials, and so on.

\item It can easily be verified that if the matrix $\matrA$ in
  Definition~\ref{def:graph:cover:construction:1} contains only zeros and
  ones, then the graph covers constructed in \GCCone\ and \GCCtwo\ never have
  parallel edges. In particular, if $\matrA$ is the parity-check matrix of a
  block code (like in most examples in this paper), then the constructed graph
  covers never have parallel edges.

  However, if $\matrA$ contains entries that are larger than one, then there
  is the potential for the constructed graph covers to have parallel edges; if
  parallel edges really appear depends then critically on the choice of the
  decomposition $\matrA = \sum_{\ell \in \setL} \matr{A}_{\ell} \text{ (in
    $\Z$)}$ and the choice of the permutation matrices $\{ \matrP_{\ell}
  \}_{\ell \in \setL}$. An example of such a case is the Tanner graph
  construction in
  Section~\ref{sec:ldpc:code:construction:Lentmaier:Kudekar:1}, where $\matrA
  \defeq \bigl[ \begin{smallmatrix} 3 & 3 \end{smallmatrix} \bigr]$ and where
  $\{ \matrA_{\ell} \}_{\ell \in \setL}$ and $\{ \matrP_{\ell} \}_{\ell \in
    \setL}$ are chosen such that parallel edges are avoided in the constructed
  graph cover.

  We note that in the case of iterated graph-cover constructions it can make
  sense to have parallel edges in the intermediate graph covers. However, in
  the last graph-cover construction stage, parallel edges are usually avoided,
  because parallel edges in Tanner graphs typically lead to a weakening of the
  code and/or of the message-passing iterative decoder.

\end{itemize}

\section{Graph-Cover Based Constructions of \\
                LDPC Convolutional Codes}
\label{sec:variations:unwrapping:1}

Although the graph-cover constructions and unwrapping techniques that were
discussed in Sections~\ref{sec:ldpc:convolutional:codes:1}
and~\ref{sec:graph:covers:1} are mathematically quite straightforward, it is
important to understand how they can be applied to obtain LDPC convolutional
codes with good performance and attractive encoder and decoder
architectures. To that end, this section explores a variety of code design
options and comments on some practical issues. It also proposes a new
``random'' unwrapping technique which leads to convolutional codes whose
performance compares favorably to other codes with the same parameters. Of
course, other variations than the ones presented here are possible, in
particular, by suitably combining some of the example constructions.

The simulation results for the codes in this section plot the decoded bit
error rate (BER) versus the signal-to-noise ratio (SNR) $E_{\mathrm{b}}/N_0$
and were obtained by assuming BPSK modulation and an additive white Gaussian
noise channel (AWGNC). All decoders were based on the sum-product
algorithm~\cite{Kschischang:Frey:Loeliger:01} and were allowed a maximum of
100 iterations, with the block code decoders employing a syndrome-check based
stopping rule. For comparing the performance of unwrapped convolutional codes
with their underlying block codes we will use the following metric.

\begin{definition}
  For a convolutional code constructed from an underlying block code, we
  define its ``convolutional gain'' to be the difference in SNR required to
  achieve a particular BER with the convolutional code compared to achieving
  the same BER with the block code. \defend
\end{definition}

The rest of this section is structured as follows. First we discuss the
construction of some \emph{time-invariant} LDPC convolutional codes based on
the Tanner unwrapping technique. In this context we make a simple observation
about how the syndrome former memory can sometimes be reduced without changing
the convolutional code. Secondly, we present a construction of
\emph{time-varying} LDPC convolutional codes based on iterated graph-cover
constructions. An important sub-topic here will be an investigation of the
influence of the ``diagonal cut'' (which is used to define a graph cover) on
the decoding performance.

\subsection{Construction of Time-Invariant LDPC Convolutional Codes 
                    Based on the Tanner Unwrapping Technique}
\label{sec:varying:circulant:size:1}

In this section we revisit a class of quasi-cyclic LDPC codes and their
associated convolutional codes that was studied
in~\cite{Smarandache:Pusane:Vontobel:Costello:09:IT}. As we will see, they are
instances of the quasi-cyclic code construction in
Example~\ref{ex:graph:cover:construction:1} and
Remark~\ref{rem:graph:cover:construction:1:variation:1}, and the corresponding
convolutional code construction based on Tanner's unwrapping technique in
Example~\ref{ex:Tanner:unwrapping:1:rev:1}.

\begin{figure} 
 \centering \psfig{figure=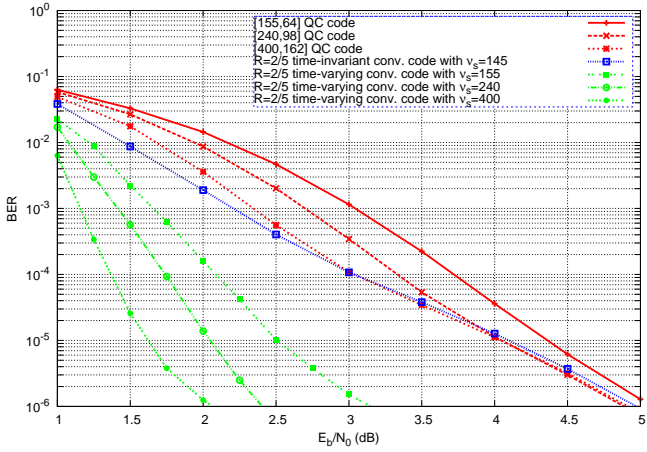, width=\columnwidth} 
 \caption{Performance of three (3,5)-regular quasi-cyclic LDPC block codes and
   their associated time-invariant and time-varying LDPC convolutional
   codes. (Note that the small gaps that appear between the second, third, and
   fourth curves for high signal-to-noise ratios are caused by a slight
   difference in code rates due to the existence of redundant rows in the
   block code parity-check matrices.)}
  \label{fig:sims} 
\end{figure}

\begin{example}
  \label{ex:Tanner:QC:code:1}

  Consider the regular proto-matrix
  \begin{align}
    \matrA
      \defeq
        \begin{bmatrix}
          1 & 1 & 1 & 1 & 1 \\
          1 & 1 & 1 & 1 & 1 \\
          1 & 1 & 1 & 1 & 1
        \end{bmatrix}
  \end{align}
  with $\mA = 3$ and $\nA = 5$. We apply \GCCone, as in
  Example~\ref{ex:graph:cover:construction:1} and
  Remark~\ref{rem:graph:cover:construction:1:variation:1}, with an interesting
  choice of permutation matrices first suggested by Tanner
  \cite{Tanner:00:ISIT} that yields the parity-check matrix
  \begin{align}
    \matrHQC{r}
      &\defeq
         \begin{bmatrix}
           \matrI_1 \hfill & \matrI_2 \hfill  & \matrI_4 \hfill &
           \matrI_8 \hfill & \matrI_{16} \hfill \\
           \matrI_5 \hfill & \matrI_{10} \hfill & \matrI_{20} \hfill &
           \matrI_9 \hfill & \matrI_{18} \hfill \\
           \matrI_{25} \hfill & \matrI_{19} \hfill & \matrI_7 \hfill &
           \matrI_{14} \hfill & \matrI_{28} \hfill
         \end{bmatrix},
           \label{eq:pcm:Tanner:QC:code:1}
  \end{align}
  where as before $\matrI_s$ is an $s$ times left-circularly shifted identity
  matrix of size $r \times r$ and $r > 28$. The corresponding polynomial
  parity-check is
  \begin{align*}
    \matrHQC{r}(X)
      &\defeq
         \begin{bmatrix}
           X^1 \hfill    & X^2 \hfill    & X^4 \hfill    &
           X^8 \hfill    & X^{16} \hfill \\
           X^5 \hfill    & X^{10} \hfill & X^{20} \hfill &
           X^9 \hfill    & X^{18} \hfill \\
           X^{25} \hfill & X^{19} \hfill & X^7 \hfill    &
           X^{14} \hfill & X^{28} \hfill
         \end{bmatrix}.
  \end{align*}
  The resulting quasi-cyclic $(3,5)$-regular LDPC block codes have block
  length $n = 5 \cdot r$. In particular, for $r = 31$, $r = 48$, and $r = 80$,
  we obtain codes of length $155$, $240$, and $400$, respectively, whose
  simulated BER performance results are shown in Figure~\ref{fig:sims}. The
  choice $r = 31$ yields the well-known length-$155$ quasi-cyclic block code
  that was first introduced by Tanner~\cite{Tanner:00:ISIT} (see also the
  discussion in~\cite{Tanner:Sridhara:Sridharan:Fuja:Costello:04:IT}).
  
  Unwrapping these codes by the Tanner unwrapping technique as in Example
  \ref{ex:Tanner:unwrapping:1:rev:1}, we obtain a rate-$2/5$ time-invariant
  convolutional code with $\nus = 145$ defined by the polynomial parity-check
  matrix
  \begin{align*}
    \matrHconv(D)
      &\defeq
         \begin{bmatrix}
           D^1 \hfill    & D^2 \hfill    & D^4 \hfill    &
           D^8 \hfill    & D^{16} \hfill \\
           D^5 \hfill    & D^{10} \hfill & D^{20} \hfill &
           D^9 \hfill    & D^{18} \hfill \\
           D^{25} \hfill & D^{19} \hfill & D^7 \hfill    &
           D^{14} \hfill & D^{28} \hfill
         \end{bmatrix}.
  \end{align*}
  Its decoding performance is also shown in Figure~\ref{fig:sims} under the
  label ``$R = 2/5$ time-invariant conv.\ code with $\nus = 145$.'' We
  conclude this example with a few remarks.
  \begin{itemize}

  \item Figure~\ref{fig:sims} shows that the convolutional code exhibits a
    ``convolutional gain'' of between $0.5\ \mathrm{dB}$ and $0.7\
    \mathrm{dB}$ compared to the $[155,64]$ quasi-cyclic LDPC block code at
    moderate BERs and that the gain remains between $0.15\ \mathrm{dB}$ and
    $0.3\ \mathrm{dB}$ at lower BERs.

  \item Note that the polynomial parity-check matrix $\matrHconv(D)$ that is
    obtained by the Tanner unwrapping technique is independent of the
    parameter $r$ of the polynomial parity-check matrix $\matrHQC{r}(X)$, as
    long as $r$ is strictly larger than the largest exponent appearing in
    $\matrHQC{r}(X)$. Moreover, for $r \to \infty$, the Tanner graph of
    $\matrHQC{r}(X)$ is closely related to the Tanner graph of
    $\matrHconv(D)$, and so it is not surprising to see that, for larger $r$,
    the decoding performance of quasi-cyclic LDPC block codes based on
    $\matrHQC{r}(X)$ tends to the decoding performance of the LDPC
    convolutional based on $\matrHconv(D)$, as illustrated by the two curves
    labeled ``$[240,98]$ QC code'' and ``$[400,162]$ QC code'' in
    Figure~\ref{fig:sims}.

  \item The permutation matrices (more precisely, the circulant matrices) that
    were used for constructing the quasi-cyclic codes in this example were
    \emph{not} chosen to optimize the Hamming distance or the pseudo-weight
    properties of the code. In particular, a different choice of circulant
    matrices may result in better high-SNR performance, i.e., in the so-called
    ``error floor'' region of the BER curve. For choices of codes with better
    Hamming distance properties, we refer the reader
    to~\cite{Hug:Bocharova:Johannesson:Kudryashov:Satyukov:10:ISIT}.
  
  \item The remaining curves in Figure~\ref{fig:sims} will be discussed in
    Example~\ref{ex:iterated:gcc:LDPC:code:construction:1}. \exend

  \end{itemize}
\end{example}

We conclude this subsection with some comments on the syndrome former memory
$\ms$ of the convolutional codes obtained by the Tanner unwrapping technique,
in particular how this syndrome former memory $\ms$ can sometimes be reduced
without changing the convolutional code.

Assume that we have obtained a polynomial parity-check matrix $\matrHconv(D)$
from $\matrHQC{r}(X)$ according to the Tanner method. Clearly, the syndrome
former memory $\ms$ is given by the largest exponent that appears in
$\matrHconv(D)$. In some instances there is a simple way of reducing $\ms$
without changing the convolutional code. Namely, if $e$ is the \emph{minimal}
exponent that appears in the polynomials of a given row of $\matrHconv(D)$,
then the polynomials in this row of $\matrHconv(D)$ can be divided by $D^e$.
We illustrate this syndrome former memory reduction for the small
convolutional code that appeared in
Example~\ref{ex:Tanner:unwrapping:1:rev:1}.

\begin{example}
  \label{ex:syndrome:former:memory:reduction:1}

  Applying the Tanner unwrapping technique to the polynomial parity-check
  matrix $\matrHQC{r}(X)$ of the quasi-cyclic LDPC code with $r = 7$ in
  Remark~\ref{rem:graph:cover:construction:1:variation:1}, we obtain
  $\matrHconv(D)$ of a rate-$1/3$ time-invariant LDPC convolutional code, as
  shown in Example~\ref{ex:Tanner:unwrapping:1:rev:1}, with syndrome former
  memory $\ms = 6$. Following the procedure discussed in the paragraph above,
  the first and second rows of $\matrHconv(D)$ can be divided by $D^1$ and
  $D^3$, respectively, to yield an equivalent convolutional code with syndrome
  former memory $\ms = 3$ and polynomial parity-check matrix
  \begin{align}
    \matrHconv(D)
      &= \begin{bmatrix}
           D^0 & D^1 & D^3 \\
           D^3 & D^2 & D^0
         \end{bmatrix}.
  \end{align}
  Figure \ref{fig:ex:syndrome:former:memory:reduction:1} shows parts of the
  corresponding scalar parity-check matrix $\matrHbarconv$ for $\ms = 3$,
  together with the original scalar parity-check matrix for $\ms = 6$, and
  illustrates the equivalence of the two matrices in the sense that only the
  ordering of the rows is different, which does not affect the corresponding
  convolutional code. In this example, the order of the even-numbered rows
  stays the same, while the odd-numbered rows are shifted by four
  positions. The equivalence of the two parity-check matrices can be seen by
  noting that the parity-check matrix, outside of the diagonal structure, is
  filled with zeros. \exend
\end{example}

\subsection{Construction of Time-Varying LDPC Convolutional Codes
                    Based on Iterated Graph-Cover Constructions}
\label{sec:time:varying:ldpc:convolutional:codes:1}

As was seen in Example~\ref{ex:graph:cover:construction:3}, interesting graph
covers can be obtained by combining \GCCone\ with \GCCtwo, or
vice-versa. Inspired by that example, this subsection considers iterated
graph-cover constructions for constructing Tanner graphs of LDPC convolutional
codes, in particular of time-varying LDPC convolutional codes.

\begin{definition}
  \label{def:iterated:gcc:LDPC:code:construction:1}

  Based on a combination of \GCCone, \GCCtwo, and the code-shortening concept
  introduced in Section~\ref{sec:unwrapping:1:rev}, we propose the following
  construction of LDPC convolutional codes.
  \begin{enumerate}

  \item We start with a proto-matrix $\matrA$ of size $\mA \times \nA$.

  \item \label{item:iterated:gcc:first:lifting:1} We apply \GCCone\ to
    $\matrA$ with finite-size permutation matrices and obtain the matrix
    $\matrA'$.

  \item \label{item:iterated:gcc:second:lifting:1} We apply \GCCtwo\ to
    $\matrA'$ with permutation matrices that are bi-infinite Toeplitz matrices
    and obtain the matrix $\matrA''$.

  \item \label{item:iterated:gcc:Tanner:unwrapping:1} Finally, looking at
    $\matrA''$ as the parity-check matrix of a bi-infinite convolutional code,
    we obtain the parity-check matrix of a convolutional code by shortening
    the code bit positions corresponding to negative time indices.

  \end{enumerate}
  Here, Steps~\ref{item:iterated:gcc:second:lifting:1}
  and~\ref{item:iterated:gcc:Tanner:unwrapping:1} can be seen as an
  application of the JFZ unwrapping method. \defend
\end{definition}

\begin{figure} 
  \begin{center}
    \includegraphics[width=0.5\textwidth]{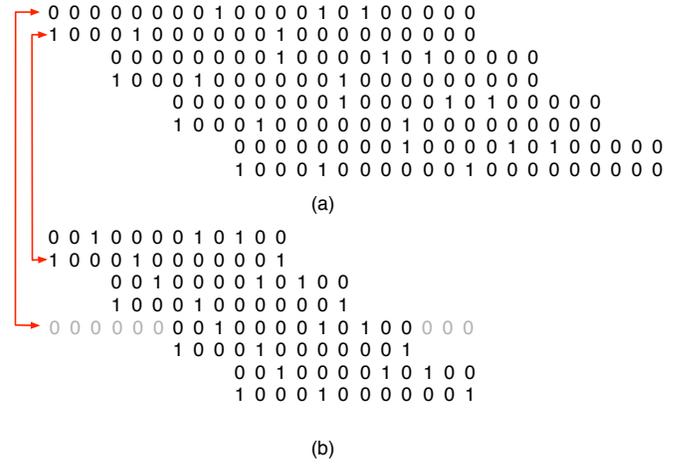}
  \end{center}
  \caption{Parts of the scalar parity-check matrices (see~\eqref{eq:TIPCC})
    corresponding to the two equivalent LDPC convolutional codes with syndrome
    former memories (a) $\ms = 6$ and (b) $\ms = 3$.}
  \label{fig:ex:syndrome:former:memory:reduction:1}
\end{figure}

The following example shows how this construction can be applied to obtain
LDPC convolutional codes with excellent performance. (In the example, where
suitable, we will refer to the analogous matrices of
Example~\ref{ex:graph:cover:construction:3} and
Figure~\ref{fig:graph:cover:construction:3} that were used to illustrate the
iterated graph-cover construction.)

\begin{example}
  \label{ex:iterated:gcc:LDPC:code:construction:1}

  Based on Definition~\ref{def:iterated:gcc:LDPC:code:construction:1}, we
  construct an LDPC convolutional code by performing the following steps.
  \begin{enumerate}

  \item We start with the same regular proto-matrix $\matrA$ as in
    Example~\ref{ex:Tanner:QC:code:1}, for which $\mA = 3$ and $\nA = 5$.

  \item We apply \GCCone\ to $\matrA$ with permutation matrices that are
    circulant matrices of size $r \times r$ and obtain the parity-check matrix
    $\matrA' = \matrHQC{r}$ shown in~\eqref{eq:pcm:Tanner:QC:code:1}, which is
    the analogue of $\matrA^{(1)}$ in
    Figure~\ref{fig:graph:cover:construction:3}(a).

  \item We apply \GCCtwo\ to $\matrA' = \matrHQC{r}$ with permutation matrices
    that are bi-infinite Toeplitz matrices and obtain a new parity-check
    matrix $\matrA''$. This is analogous to the transition of the matrix
    $\matrA^{(1)}$ in Figure~\ref{fig:graph:cover:construction:3}(a) to the
    matrix $\omatrB^{(1)}$ in
    Figure~\ref{fig:graph:cover:construction:3}(c). The ``diagonal cut'' is
    obtained by alternately moving $\nA = 5$ units to the right and then $\mA
    = 3$ units down.

  \item Finally, we obtain the desired convolutional code by shortening the
    code bit positions corresponding to negative time indices.

  \end{enumerate}
  For the choices $r = 31$, $48$, $80$, this construction results in
  rate-$2/5$ time-varying convolutional codes with syndrome former memory $\ms
  = 30$, $47$, $79$, respectively, and with constraint length $\nus = (\ms+1)
  \cdot \nA = 155$, $240$, $400$, respectively. The label ``time-varying'' is
  indeed justified because the convolutional codes constructed here can be
  expressed in the form of the parity-check matrix in~\eqref{eq:PCC} with a
  suitable choice of syndrome former memory $\ms$, non-trivial period $\Ts$,
  and submatrices $\bigl\{ \matrH_i(t) \bigr\}_i$.

  The decoding performance of these codes is shown in Figure~\ref{fig:sims},
  labeled ``R = 2/5 time-varying conv.\ code with $\nus = \ldots$''. As
  originally noted in \cite{Pusane:Smarandache:Vontobel:Costello:07:ISIT}, we
  observe that these three LDPC convolutional codes achieve significantly
  better performance at a BER of $10^{-6}$ than the other codes shown in this
  plot, namely with ``convolutional gains'' of $2.0\ \mathrm{dB}$ for the
  $\nus = 155$ convolutional code, $2.4\ \mathrm{dB}$ for the $\nus = 240$
  convolutional code, and $2.8\ \mathrm{dB}$ for the $\nus = 400$
  convolutional code, compared to the three respective underlying LDPC block
  codes.

  In order to compare these codes based on a given decoding processor
  (hardware) complexity, we consider a block code of length $n = \nus$
  (see~\cite{Costello:Pusane:Bates:Zigangirov:06:ITA}
  and~\cite{Costello:Pusane:Jones:Divsalar:07:ITA}). The above time-varying
  convolutional code for $r = 31$ has constraint length $\nus = (\ms+1) \cdot
  c = 155$, and hence approximately the same processor complexity as the
  quasi-cyclic block code of length $n = 155$ in Figure~\ref{fig:sims} and the
  time-invariant convolutional code with $\nus = 145$ in
  Figure~\ref{fig:sims}, but it achieves large gains compared to both of these
  codes. We note, in addition, that the performance of the time-varying
  convolutional code with constraint length $\nus = 400$ is quite remarkable,
  since, at a BER of $10^{-5}$, it performs within $1\ \mathrm{dB}$ of the
  iterative decoding threshold of $0.965\ \mathrm{dB}$, while having the same
  processor complexity as a block code of length only $n = 400$. In Section
  \ref{sec:cost}, we discuss some possible reasons for these ``convolutional
  gains,'' along with their associated implementation costs in terms of
  decoder memory and decoding delay. \exend
\end{example}

We make the following observations with respect to the above definition and
example.

\begin{itemize}

\item The LDPC code construction in the above example yields time-varying LDPC
  convolutional codes with syndrome former memory $\ms \leq r - 1$ and period
  $\Ts = r$. Most importantly, varying $r$ in the above construction leads to
  different LDPC convolutional codes. This is in contrast to the Tanner
  unwrapping technique discussed in
  Section~\ref{sec:varying:circulant:size:1}, where the obtained LDPC
  convolutional code is independent of the parameter $r$, as long as $r$ is
  strictly larger than the largest exponent in $\matrHQC{r}(X)$.

\item As mentioned previously in Example~\ref{ex:graph:cover:construction:3},
  the iterated graph-cover construction based on the combination of \GCCone\
  and \GCCtwo\ yields Tanner graphs that have a ``pseudo-random'' structure, a
  structure that seems to be beneficial as indicated by the above simulation
  results. (We remark that the improved performance of the time-varying LDPC
  convolutional codes obtained by unwrapping a randomly constructed LDPC block
  code was first noted by Lentmaier \textit{et al.}
  in~\cite{Lentmaier:Truhachev:Zigangirov:01:Problems}.)

\item Instead of constructing a first parity-check matrix as in
  Step~\ref{item:iterated:gcc:first:lifting:1} of
  Definition~\ref{def:iterated:gcc:LDPC:code:construction:1}, one can also
  start with any other (randomly or non-randomly constructed, regular or
  irregular) parity-check matrix, and still achieve a ``convolutional gain.''
  The next example is an illustration of this point.

\end{itemize}

\begin{figure} 
  \centering \includegraphics[width=\columnwidth]
                             {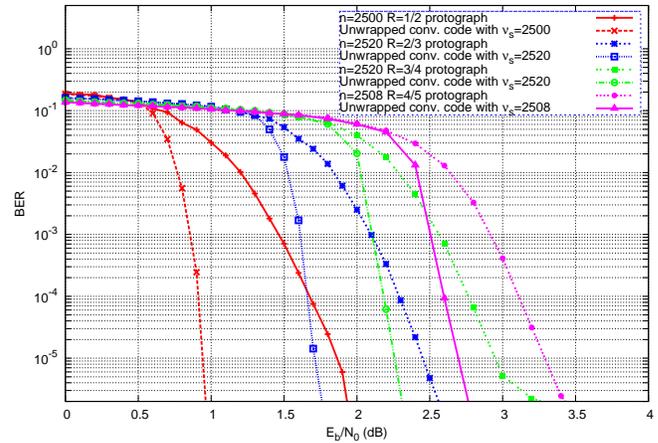} 
  \caption{Performance of a family of irregular proto-graph-based
    LDPC block codes and the associated time-varying LDPC convolutional
    codes.}\label{fig:JPLsims}
\end{figure}

\begin{example}
  \label{example:divsalar:codes:derivatives:1}

  As was done in~\cite{Costello:Pusane:Jones:Divsalar:07:ITA}, one can replace
  the parity-check matrix that was constructed in
  Step~\ref{item:iterated:gcc:first:lifting:1} of
  Definition~\ref{def:iterated:gcc:LDPC:code:construction:1} by an irregular
  LDPC block code with optimized iterative decoding thresholds. In particular,
  one can start with the parity-check matrix of the rate-$1/2$ irregular
  proto-graph-based code from~\cite{Divsalar:Jones:Dolinar:Thorpe:05:GLOBECOM}
  with an iterative decoding threshold of $0.63\ \mathrm{dB}$, and several of
  its punctured versions. Figure~\ref{fig:JPLsims} shows simulation results
  for the obtained block and convolutional codes. Each simulated block code
  had a block length of about $2500$, with code rates ranging from $1/2$ to
  $4/5$. We see that ``convolutional gains'' ranging from $0.6\ \mathrm{dB}$
  to $0.9\ \mathrm{dB}$ at a BER of $10^{-5}$ were obtained.

  Similarly, it was shown in~\cite{Pusane:Zigangirov:Costello:06:ICC} that an
  LDPC convolutional code derived from a randomly constructed rate-$1/2$
  irregular LDPC block code with block length $2400$ outperformed the
  underlying code by almost $0.8\ \mathrm{dB}$ at a BER of $10^{-5}$. The
  degree distribution of the underlying LDPC block code was fully optimized
  and had an iterative decoding threshold of $0.3104\
  \mathrm{dB}$~\cite{Richardson:Shokrollahi:Urbanke:01:IT}. \exend
\end{example}

Of course, there are other ways of applying the ``diagonal cut'' in
Step~\ref{item:iterated:gcc:second:lifting:1} of
Example~\ref{ex:iterated:gcc:LDPC:code:construction:1}, and so it is
natural to investigate the influence of different ``diagonal cuts'' on the
decoding performance. We will do this in the next few paragraphs by extending
the discussion that was presented right after Example~\ref{example:unwrap:1}.

We start by assuming that the matrix after
Step~\ref{item:iterated:gcc:first:lifting:1} of
Definition~\ref{def:iterated:gcc:LDPC:code:construction:1} has size $m \times
n$, and define $\eta \defeq \gcd(m, n)$. Then, for any positive integer $\ell$
that divides $\eta$, we can perform a ``diagonal cut'' where we alternately
move $c' = \ell \cdot (n/\eta)$ units to the right and then $c' - b' \defeq
\ell \cdot (m/\eta)$ units down (i.e., $b' = \ell \cdot \bigl( (n-m)/\eta
\bigr)$. With this, the obtained convolutional code is a periodically
time-varying LDPC convolutional code with rate $R' = b' / c' = b / c$,
syndrome former memory $\ms' = (n/c') - 1 = (\eta/\ell) - 1$, period $\Ts' =
\ms' + 1 = n / c' = \eta / \ell$, and constraint length $\nus' = c' \cdot
(\ms' + 1) = n$. (Note that the syndrome former memory $\ms'$ depends on
$\ell$, but the constraint length $\nus'$ is independent of $\ell$.)

\begin{figure} 
  \begin{center}
\psfrag{label01}{\tiny \sf Block code with n=2048}
\psfrag{label02}{\tiny \sf Conv.\ code with
b'=$\ell$=1, c'=2, m'$_{\!\textrm{\sf s}}$=1023}
\psfrag{label03}{\tiny \sf Conv.\ code with
  b'=$\ell$=2, c'=4, m'$_{\!\textrm{\sf s}}$=511}
\psfrag{label04}{\tiny \sf Conv.\ code with
  b'=$\ell$=4, c'=8, m'$_{\!\textrm{\sf s}}$=255}
\psfrag{label05}{\tiny \sf Conv.\ code with
  b'=$\ell$=8, c'=16, m'$_{\!\textrm{\sf s}}$=127}
\psfrag{label06}{\tiny \sf Conv.\ code with
  b'=$\ell$=16, c'=32, m'$_{\!\textrm{\sf s}}$=63}
\psfrag{label07}{\tiny \sf Conv.\ code with 
  b'=$\ell$=32, c'=64, m'$_{\!\textrm{\sf s}}$=31}
\psfrag{label08}{\tiny \sf Conv.\ code with 
  b'=$\ell$=64, c'=128, m'$_{\!\textrm{\sf s}}$=15}
\psfrag{label09}{\tiny \sf Conv.\ code with 
  b'=$\ell$=128, c'=256, m'$_{\!\textrm{\sf s}}$=7}
\psfrag{label10}{\tiny \sf Conv.\ code with 
  b'=$\ell$=256, c'=512, m'$_{\!\textrm{\sf s}}$=3}
\psfrag{label11}{\tiny \sf Conv.\ code with 
  b'=$\ell$=512, c'=1024, m'$_{\!\textrm{\sf s}}$=1}
\psfrag{label12}{\tiny \sf Conv.\ code with 
  b'=$\ell$=1024, c'=2048, m'$_{\!\textrm{\sf s}}$=0}
     \includegraphics[width=\columnwidth]{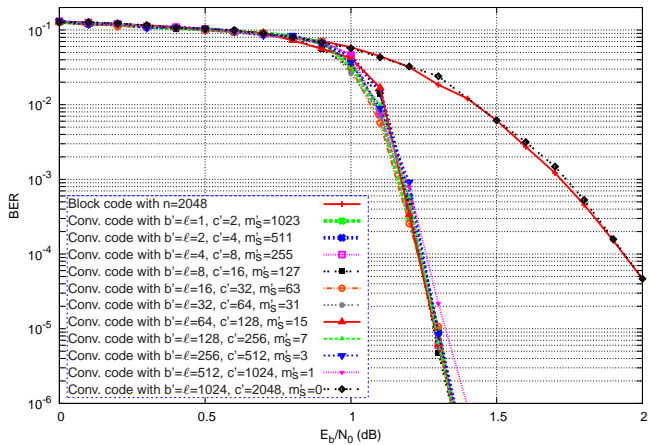}
  \end{center} 
  \caption{Performance of a family of LDPC convolutional codes obtained from a
    $(3,6)$-regular LDPC block code using different step sizes.}
  \label{fig:results400}
\end{figure}

\begin{example}
  \label{ex:diagonal:cuts:big:ldpc:block:codes:1}

  Here we simulate the performance of some LDPC convolutional codes obtained
  according to the above generalization of the ``diagonal cut.'' Namely, we
  start with a randomly-constructed $(3,6)$-regular LDPC block code based on a
  parity-check matrix of size $1024 \times 2048$. Therefore $m = 1024$, $n =
  2048$, and $\eta \defeq \gcd(m, n) = 1024$. (Note that $c' = \ell \cdot
  (n/\eta) = 2\ell$ and $b' = \ell \cdot \bigl( (n-m)/\eta \bigr) = \ell$ in
  this case.)  Figure \ref{fig:results400} shows the performance of the
  resulting family of LDPC convolutional codes, where $\ell$ varies in powers
  of $2$ from $1$ to $1024$, each with constraint length $\nus' = 2048$. We
  make the following observations. First, the case $\ell = 1024$ is not
  interesting because it results in $\ms' = 0$, i.e., it is a trivial
  concatenation of copies of the block code, and so the BER is the same as for
  the underlying block code. Secondly, for all other choices of $\ell$, the
  constructed codes perform very similarly, each exhibiting a sizable
  ``convolutional gain'' compared to the block code, although the syndrome
  former memory $\ms'$ is different in each case. \exend
\end{example}

A special case of the above code construction deserves mention. When $\eta =
1$, i.e., $m$ and $n$ are relatively prime, the only possible step size is
obtained by choosing $\ell = \eta = 1$, which results in the above-mentioned
uninteresting case of trivial concatenations of copies of the block
code. However, all-zero columns can be inserted in the parity-check matrix
such that a value of $\eta > 1$ is obtained, which allows a step size to be
chosen that results in a convolutional code with $\ms' > 0$.  The variable
nodes corresponding to the all-zero columns are not transmitted, i.e., they
are punctured, so that the rate corresponds to the size of the original
parity-check matrix.

For the ``diagonal cut'' LDPC convolutional code constructions discussed
above, the unwrapped convolutional codes have the minimum possible constraint
length $\nus'$, which is equal to the block length of the underlying block
code. Although this is a desirable property for practical implementation, we
do not need to limit ourselves to diagonal cuts in general.

Inspired by the graph-cover construction of
Figures~\ref{fig:graph:cover:constructions:2}(b)
and~\ref{fig:graph:cover:constructions:2}(d) in
Example~\ref{ex:graph:cover:construction:2}, instead of a ``diagonal cut'' we
now consider a ``random cut,'' which we define as a partition of the
parity-check matrix into two matrices that add up (over $\Z$) to the
parity-check matrix. Despite the randomness of this approach, several of the
key unwrapping properties of the ``diagonal cut'' are preserved. For example,
the computational complexity per decoded bit does not change, since the degree
distributions of the resulting codes are all equal.\footnote{This is
  guaranteed by choosing a random partition of the block code parity-check
  matrix and then using this partition to construct one period of the
  time-varying convolutional code parity-check matrix.} However, the LDPC
convolutional codes based on a ``random cut'' typically require larger
decoding processor sizes as a result of increased code constraint lengths.

\begin{figure} 
\psfrag{label01}{\tiny \sf Block code with n=2048}
\psfrag{label02}{\tiny \sf Conv.\ code with b'=1, c'=2 ($\ell$=1)}
\psfrag{label03}{\tiny \sf Conv.\ code with random partition}
  \centering \includegraphics[width=\columnwidth]{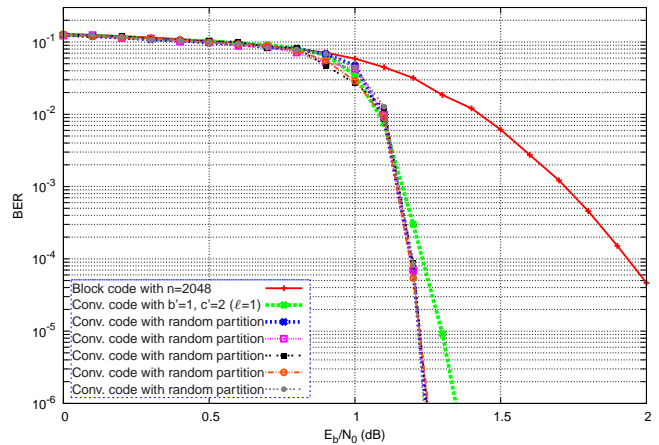} 
  \caption{Performance of ``randomly unwrapped'' LDPC convolutional codes
    obtained from a $(3,6)$-regular LDPC block code using random
    partitions.}
  \label{fig:partition}
\end{figure}

\begin{example}
  \label{ex:random:cuts:big:ldpc:block:codes:1}

  We continue Example~\ref{ex:diagonal:cuts:big:ldpc:block:codes:1}; however,
  instead of performing ``diagonal cuts,'' we perform ``random cuts.'' Figure
  \ref{fig:partition} shows the performance of five such LDPC convolutional
  codes, each with rate $1/2$ and constraint length $\nus' = 4096$, compared
  to the underlying block code and the LDPC convolutional code constructed in
  Example~\ref{ex:diagonal:cuts:big:ldpc:block:codes:1} (with parameters $\ell
  = 1$, $b' = 1$, $c' = 2$, and $\nus' = 2048$). We note that the increase in
  constraint length from $\nus' = 2048$ to $\nus' = 4096$ due to the ``random
  cut'' results in a small additional coding gain in exchange for the
  larger decoding processor size. \exend
\end{example}

Finally, we note that, for a size $m\times n$ sparse parity-check matrix
$\matr{H}$ with $p$ nonzero entries, there are a total of $2^{mn}$ possible
ways of choosing a random cut. However, due to the sparsity, there are only
$2^{p}$ distinct random cuts, where $p \ll m \cdot n$.

\section{Connections to Other LDPC Codes Based on
                Graph-Cover Constructions}
\label{sec:other:LDPC:code:constructions:1}

In this section we briefly discuss some other graph-cover-based LDPC code
constructions proposed in the literature, namely by Ivkovic \textit{et
  al.}~\cite{Ivkovic:Chilappagari:Vasic:08:1}, by Divsalar \textit{et
  al.}~\cite{Divsalar:Jones:Dolinar:Thorpe:05:GLOBECOM, CCSDS:07:1}, by
Lentmaier \textit{et al.}\cite{Lentmaier:Fettweis:Zigangirov:Costello:09:ITA,
  Lentmaier:Mitchell:Fettweis:Costello:10:ITA}, and by Kudekar \textit{et
  al.}~\cite{Kudekar:Richardson:Urbanke:10:online}.

\subsection{LDPC Code Construction by Ivkovic \textit{et al.}}

The LDPC code construction by Ivkovic \textit{et al.}\
in~\cite{Ivkovic:Chilappagari:Vasic:08:1} can be seen as an application of the
graph-cover construction in Figures~\ref{fig:graph:cover:constructions:2}(b)
and~\ref{fig:graph:cover:constructions:2}(d) in
Example~\ref{ex:graph:cover:construction:2}. Namely, in terms of our notation,
Ivkovic \textit{et al.}~\cite{Ivkovic:Chilappagari:Vasic:08:1} start with a
parity-check matrix $\matrH$, choose the set $\setL \defeq \{ 0, 1 \}$, a
collection of zero-one matrices $\{ \matrH_0, \matrH_1 \}$ such that $\matrH =
\matrH_0 + \matrH_1 \ (\text{in $\Z$})$, and the collection of permutation
matrices $\{ \matrP_0, \matrP_1 \}_{\ell \in \setL} \defeq \bigl\{
\bigl[ \begin{smallmatrix} 1 & 0 \\ 0 & 1 \end{smallmatrix} \bigr],
\bigl[ \begin{smallmatrix} 0 & 1 \\ 1 & 0 \end{smallmatrix} \bigr]
\bigr\}$. Most importantly, the decomposition of $\matrH$ into $\matrH_0$ and
$\matrH_1$ is done such that trapping sets that were present in the Tanner
graph of $\matrH$ are not present in the Tanner graph of the new parity-check
matrix. In addition, Ivkovic \textit{et al.} give guarantees on the
relationship between the minimum Hamming distances of the old and new
code.\footnote{See also the discussion of similar results
  in~\cite[Appendix~J]{Smarandache:Vontobel:09:1:subm}.}

\subsection{LDPC Code Construction by Divsalar \textit{et al.}}

One of the LDPC code constructions by Divsalar \textit{et al.}\
in~\cite{Divsalar:Jones:Dolinar:Thorpe:05:GLOBECOM, CCSDS:07:1} is the
so-called rate-$1/2$ AR4JA LDPC code construction, which was also considered
earlier in Example~\ref{example:divsalar:codes:derivatives:1}. A particularly
attractive, from an implementation perspective, version of this code
construction is obtained by an iterated graph-cover construction procedure,
where each graph-cover construction is based on a cyclic cover, as in the
application of \GCCone\ in Example~\ref{ex:graph:cover:construction:1}.
Although cyclic covers result in simplified encoding and decoding circuitry,
codes based on cyclic covers are known to have the disadvantage that the
minimum Hamming distance is upper bounded by a number that is a function of
the proto-graph structure~\cite{Smarandache:Vontobel:09:1:subm,
  Butler:Siegel:10:1}. However, because the cyclic cover of a cyclic cover of
the proto-graph is \emph{not necessarily} a cyclic cover of the proto-graph,
such disadvantages are avoided to a certain extent in the AR4JA LDPC code
construction. Nevertheless, ultimately the minimum Hamming distance of such
codes will also be upper bounded by some number; however, these bounds usually
become relevant only beyond the code length of interest.\footnote{For this
  statement we assume that the degree of the first cover is fixed.}

\subsection{LDPC Code Construction by Lentmaier \textit{et al.} and 
                    Kudekar \textit{et al.}}
\label{sec:ldpc:code:construction:Lentmaier:Kudekar:1}

The LDPC code constructions by Lentmaier \textit{et
  al.}~\cite{Lentmaier:Fettweis:Zigangirov:Costello:09:ITA,
  Lentmaier:Mitchell:Fettweis:Costello:10:ITA} and Kudekar \textit{et
  al.}~\cite{Kudekar:Richardson:Urbanke:10:online} can also be seen as
iterated graph-cover constructions. We now describe a specific instance of
this construction.
\begin{itemize}

\item It starts with a proto-matrix $\matrA \defeq \bigl[ \begin{smallmatrix}
    3 & 3 \end{smallmatrix} \bigr]$.

\item The first graph-cover construction is very similar to the bi-infinite
  graph-cover construction in
  Example~\ref{ex:graph:cover:construction:1:infinite:1} and
  Figure~\ref{fig:graph:cover:constructions:1:variation:2}. Namely, in terms
  of our notation, we define the set $\setL \defeq \{ 0, 1, 2, 3, 4, 5 \}$,
  the collection of matrices $\{ \matrA_{\ell} \}_{\ell \in \setL}$ with
  $\matrA_0 = \matrA_1 = \matrA_2 = \bigl[ \begin{smallmatrix} 1 &
    0 \end{smallmatrix} \bigr]$ and $\matrA_3 = \matrA_4 = \matrA_5 =
  \bigl[ \begin{smallmatrix} 0 & 1 \end{smallmatrix} \bigr]$, and the
  collection of permutation matrices $\{ \matrP_{\ell} \}_{\ell \in \setL}$
  with $\matrP_0 \defeq \matrT_0$, $\matrP_1 \defeq \matrT_1$, $\matrP_2
  \defeq \matrT_2$, $\matrP_3 \defeq \matrT_0$, $\matrP_4 \defeq \matrT_1$,
  $\matrP_5 \defeq \matrT_2$, where as before $\matrT_s$ is a bi-infinite
  Toeplitz matrix with zeros everywhere except for ones in the $s$-th diagonal
  below the main diagonal.

\item The second graph-cover construction is a random graph-cover construction
  of cover-degree $M$.

\item The code is shortened. Namely, for some positive integer $L$ all
  codeword indices corresponding to values outside the range $[-LM,LM]$ are
  shortened.\footnote{Although this code construction method could be
    presented such that the shortening is done between the two graph-cover
    construction steps, namely by shortening all codeword indices that
    correspond to values outside the range $[-L,L]$, we have opted to present
    the code construction such that the shortening is done after the two
    graph-cover construction steps. In this way, the structure of the code
    construction description matches better the description in
    Definition~\ref{def:iterated:gcc:LDPC:code:construction:1}.}

\end{itemize}
We now point out some differences between this code construction and the LDPC
convolutional code construction in
Definition~\ref{def:iterated:gcc:LDPC:code:construction:1}. Namely, the LDPC
code ensemble constructed above has the following properties.
\begin{itemize}

\item The first graph-cover construction is based on bi-infinite Toeplitz
  permutation matrices, and the second graph-cover construction is based on
  finite-size permutation matrices.

\item The analysis focuses on the case where $M$ and $L$ go to infinity (in
  that order), i.e., for a fixed $L$ the parameter $M$ tends to
  infinity. Afterwards, $L$ tends to infinity.

\item The number of check nodes with degree smaller than $6$ in the Tanner
  graph is proportional to $M$.

\item In~\cite{Kudekar:Richardson:Urbanke:10:online}, for the binary erasure
  channel, when $M$ and $L$ go to infinity (in that order), Kudekar \textit{et
    al.} prove that the sum-product algorithm decoding threshold for a slight
  variation of the above-mentioned ensemble of codes equals the maximum
  a-posteriori decoding threshold for the ensemble of $(3,6)$-regular LDPC
  codes. This is a very remarkable property!
  (In~\cite{Lentmaier:Sridharan:Costello:Zigangirov:10:IT:appear}, using
  density evolution methods, Lentmaier \textit{et al.} give numerical evidence
  that this statement might also hold for binary-input output-symmetric
  channels beyond the binary erasure channel.)

\end{itemize}
On the other hand, the codes constructed in
Definition~\ref{def:iterated:gcc:LDPC:code:construction:1} have the following
properties. (We assume that the underlying block code is a $(3,6)$-regular
LDPC code.)
\begin{itemize}

\item The first graph-cover construction is based on finite-size permutation
  matrices, and the second graph-cover construction is based on bi-infinite
  Toeplitz permutation matrices.

\item In a typical application of this construction, $r$ is fixed.

\item The number of check nodes with degree smaller than $6$ in the Tanner
  graph of the LDPC convolutional code is proportional to $r$.

\item For a binary-input output-symmetric channel, the performance of the
  unterminated LDPC convolutional code under the continuous sliding window
  sum-product algorithm decoding discussed in Section \ref{sec:implement}
  improves with increasing $r$ (see, e.g., Fig. \ref{fig:sims}), but the
  ultimate asymptotic threshold of such unterminated decoding is
  unknown.\footnote{Lentmaier \textit{et al.} have shown
    in~\cite{Lentmaier:Fettweis:Zigangirov:Costello:09:ITA}
    and~\cite{Lentmaier:Mitchell:Fettweis:Costello:10:ITA} that properly
    terminated LDPC convolutional codes become equivalent to the LDPC block
    codes constructed by Kudekar \textit{et al.}
    in~\cite{Kudekar:Richardson:Urbanke:10:online} and inherit their excellent
    asymptotic threshold properties, but whether this is true for unterminated
    LDPC convolutional codes is still an open question.}

\end{itemize}
The differences between these two code families come mainly from the fact that
the codes constructed by Lentmaier \textit{et al.} and Kudekar \textit{et al.}
are essentially block codes, although sophisticated ones, whereas the codes in
Definition~\ref{def:iterated:gcc:LDPC:code:construction:1} are convolutional
codes, along with their advantages and disadvantages. In particular, the way
the limits of the parameters are taken, there is a significant difference in
the fraction of check nodes with degree strictly smaller than $6$. Namely, in
the case of the codes by Lentmaier \textit{et al.} and Kudekar \textit{et al.}
this fraction is a fixed non-zero function of $L$ (here we assume fixed $L$
and $M \to \infty$), whereas in the case of the codes considered in this
paper, this fraction is zero (here we assume fixed $r$ and an unterminated
convolutional code).

We conclude this section with the following remarks. Namely, although the
convolutional codes in
Definition~\ref{def:iterated:gcc:LDPC:code:construction:1} may not enjoy the
same asymptotic thresholds as the block code constructions by Lentmaier
\textit{et al.} and by Kudekar \textit{et al.}, they lend themselves to a
continuous decoding architecture, as described in Section \ref{sec:implement},
which can be advantageous in certain applications, such as data streaming,
without a predetermined frame structure. More importantly, however, it is very
encouraging that the simulation results reported in this paper indicate that
sizable ``convolutional gains'' are already visible for very reasonable
constraint/code lengths. In the next section we discuss some possible reasons
for these gains. Finally, it is worth noting that, as the block lengths and
associated constraint lengths of the constructions presented in this section
become larger, the observed ``convolutional gains'' will become smaller since
the block code results will approach their respective thresholds.

\section{Analysis of Derived \\ LDPC Convolutional Codes}
\label{sec:analysis}

This section collects some analytical results about LDPC convolutional
codes. In particular, we compare the existence~/ non-existence of cycles in
LDPC block and LDPC convolutional codes, we present some properties of
pseudo-codewords, and we discuss the --- mostly moderate --- cost increase in
decoder complexity that is incurred by going from LDPC block to LDPC
convolutional codes.

\subsection{Graph-Cycle Analysis}
\label{sec:cycleanalysis}

It is well known that cycles in the Tanner graph representation of a sparse
code affect message-passing iterative decoding algorithms, with short cycles
generally pushing the performance further away from optimum. (Indeed, attempts
to investigate and minimize these effects have been made in
\cite{Tian:Jones:Villasenor:Wesel:04:TCOM} and
\cite{Ramamoorthy:Wesel:04:ISIT}, where the authors propose LDPC code
construction procedures to maximize the connectivity of short cycles to the
rest of the graph, thus also maximizing the independence of the messages
flowing through a cycle.) Hence it is common practice to design codes that do
not contain short cycles, so as to obtain independent messages in at least the
initial iterations of the decoding process.

Avoiding cycles in Tanner graphs also has the benefit of avoiding
pseudo-codewords.\footnote{Here and in the following, pseudo-codewords refer
  to pseudo-codewords as they appear in linear programming (LP)
  decoding~\cite{Feldman:03:diss, Feldman:Wainwright:Karger:05:IT} and in the
  graph-cover-based analysis of message-passing iterative decoding
  in~\cite{Koetter:Vontobel:03:ISTC, Vontobel:Koetter:05:IT:subm}. For other
  notions of pseudo-codewords, in particular computation tree
  pseudo-codewords, we refer to the discussion
  in~\cite{Axvig:Dreher:Morrison:Psota:Perez:Walker:09:IT}.} To see this, let
the active part of a pseudo-codeword be defined as the set of bit nodes
corresponding to the support of the pseudo-codeword, along with the adjacent
edges and check nodes. With this, it holds that the active part of any
pseudo-codeword contains at least one cycle and/or at least one bit node of
degree one. And so, given that the typical Tanner graph under consideration in
this paper does not contain bit nodes of degree one, the active part of a
pseudo-codeword must contain at least one cycle. Therefore, avoiding cycles
implicitly means avoiding pseudo-codewords.\footnote{Note that the support of
  any pseudo-codeword is a stopping set~\cite{Koetter:Vontobel:03:ISTC,
    Vontobel:Koetter:05:IT:subm, Kelley:Sridhara:07:1}.}

Let $\tmatrH$ and $\matrH$ be two parity-check matrices such that
$\graphT(\tmatrH)$ is a graph cover of $\graphT(\matrH)$. It is a well-known
result that any cycle in $\graphT(\tmatrH)$ can be mapped into a cycle in
$\graphT(\matrH)$. This has several consequences. In particular, the girth of
$\graphT(\tmatrH)$ is at least as large as the girth of $\graphT(\matrH)$, and
more generally, $\graphT(\tmatrH)$ contains fewer short cycles than
$\graphT(\matrH)$.\footnote{This observation has been used in many different
  contexts over the past ten years in the construction of LDPC and turbo
  codes; in particular, it was used
  in~\cite{Lentmaier:Truhachev:Zigangirov:01:Problems}, where the authors
  dealt with bounding the girth of the resulting LDPC convolutional codes.}
For the codes constructed in this paper, this means that the unwrapping
process (from block code to convolutional code) can ``break'' some cycles in
the Tanner graph of the block code.

We now revisit some codes that where discussed in earlier sections and analyze
their graph cycle structure using a brute-force search algorithm.\footnote{The
  search technique that we used is based on evaluating the diagonal entries of
  the powers of the matrix $\matrM$ defined
  in~\cite[Eq.~(3.1)]{Stark:Terras:96:1}. Note that this search technique
  works only for counting cycles of length smaller than twice the girth of the
  graph. For searching longer cycles, more sophisticated algorithms are
  needed.} Note that, in order to accurately compare the graph cycle
distributions of two codes with different block/constraint lengths, we compute
the total number of cycles of a given cycle length per block/constraint
length, and divide this number by the block/constraint length.\footnote{For
  LDPC convolutional codes, we have made use of the periodicity of the
  parity-check matrices in order to complete the search in a finite number of
  steps.}

\begin{table}
  \caption{Average (per bit node) number $\bar N_{\ell}$ of cycles 
    of length $\ell$ for the Tanner graphs of the block codes (BCs) of block
    length $n$ and convolutional codes (CCs) of constraint length $\nus$ 
    discussed in Example~\ref{ex:cyc1}. (All Tanner graphs have girth $8$.)}
  \label{table:girth:average:cycle:lengths:example:1}

  \begin{center}
    \begin{tabular}{|l|c|c|c|c|}
      \hline
             &            &              &               \\[-0.26cm]
        Code & $\bar N_8$ & $\bar N_{10}$ & $\bar N_{12}$ \\
      \hline
      \hline
        BC ($n = 155$)
          & $3.000$ & $24.000$ & $146.000$ \\
      \hline
        BC ($n = 240$)
          & $2.600$ & $14.000$ & $93.400$ \\
      \hline
        BC ($n = 400$)
          & $2.200$ & $12.400$ & $70.600$ \\
      \hline
      \hline
        Time-invariant CC ($\nus = 145$)
          & $2.200$ & $12.400$ & $70.200$ \\
      \hline
      \hline
        Time-varying CC ($\nus = 155$)
          & $0.910$ & $8.342$ & $44.813$ \\
      \hline
        Time-varying CC ($\nus = 240$)
          & $0.917$ & $5.338$ & $30.242$ \\
      \hline
        Time-varying CC ($\nus = 400$)
          & $0.675$ & $4.705$ & $24.585$ \\
      \hline
    \end{tabular}
  \end{center}
\end{table}

\begin{example}
  \label{ex:cyc1}

  Consider the LDPC block and convolutional codes that were constructed in
  Examples~\ref{ex:Tanner:QC:code:1}
  and~\ref{ex:iterated:gcc:LDPC:code:construction:1} and whose BER performance
  was plotted in
  Figure~\ref{fig:sims}. Table~\ref{table:girth:average:cycle:lengths:example:1}
  shows the average number of cycles of certain lengths for the Tanner graphs
  of the quasi-cyclic block codes, for the Tanner graph of the corresponding
  time-invariant convolutional code, and for the Tanner graph of the
  time-varying convolutional codes. \exend
\end{example}

\begin{table}
  \caption{Average (per bit node) number $\bar N_{\ell}$ of cycles 
    of length $\ell$ for the Tanner graphs of the block codes (BCs) of block 
    length $n$ and convolutional codes (CCs) of constraint length $\nus$
    discussed in Example~\ref{ex:cyc3}. (All Tanner graphs have girth $4$.)}
  \label{table:girth:average:cycle:lengths:example:2}

  \begin{center}
    \begin{tabular}{|l|c|c|}
      \hline
             &            &            \\[-0.26cm] 
        Code & $\bar N_4$ & $\bar N_6$ \\ 
      \hline
      \hline
        Rate-$1/2$ BC ($n = 2500$)
          & $0.013$ & $0.120$ \\ 
      \hline
        Rate-$2/3$ BC ($n = 2520$)
 	  & $0.065$ & $0.839$ \\ 
      \hline
        Rate-$3/4$ BC ($n = 2520$)
          & $0.136$ & $2.710$ \\ 
      \hline
        Rate-$4/5$ BC ($n = 2508$)
 	  & $0.250$ & $6.544$ \\ 
      \hline
      \hline
        Rate-$1/2$ time-varying CC ($\nus = 2500$)
          & $0.010$ & $0.064$ \\ 
      \hline
        Rate-$2/3$ time-varying CC ($\nus = 2520$)
 	  & $0.044$ & $0.483$ \\ 
      \hline
        Rate-$3/4$ time-varying CC ($\nus = 2520$)
          & $0.091$ & $1.465$ \\ 
      \hline
        Rate-$4/5$ time-varying CC ($\nus = 2508$)
 	  & $0.173$ & $3.622$ \\ 
      \hline
    \end{tabular}
  \end{center}
\end{table}

\begin{example}
  \label{ex:cyc3}

  Table~\ref{table:girth:average:cycle:lengths:example:2} shows the cycle
  analysis results for the rate-$1/2$ proto-graph-based codes that were
  discussed in Example~\ref{example:divsalar:codes:derivatives:1} and whose
  BER performance was plotted in Figure~\ref{fig:JPLsims}. \exend
\end{example}

From Examples \ref{ex:cyc1} and~\ref{ex:cyc3}, we see that many of the short
cycles in the Tanner graphs of the LDPC block codes are ``broken'' to yield
cycles of larger length in the Tanner graphs of the derived LDPC convolutional
codes.

\subsection{Pseudo-Codeword Analysis }
\label{sec:PCanalysis}

This section collects some comments concerning the pseudo-codewords of the
parity-check matrices under consideration in this paper.

We start by observing that many of the statements that were made
in~\cite{Smarandache:Pusane:Vontobel:Costello:09:IT} about pseudo-codewords
can be extended to the setup of this paper. In particular, if some
parity-check matrices $\tmatrH$ and $\matrH$ are such that $\graphT(\tmatrH)$
is a graph cover of $\graphT(\matrH)$, then a pseudo-codeword of $\tmatrH$ can
be ``wrapped'' to obtain a pseudo-codeword of $\matrH$, as is formalized in
the next lemma.

\begin{lemma}
  Let the parity-check matrices $\tmatrH$ and $\matrH$ be such that
  $\graphT(\tmatrH)$ is an $M$-fold graph cover of $\graphT(\matrH)$. More
  precisely, let $\tmatrH = \sum_{\ell \in \setL} \matrH_{\ell} \otimes
  \matrP_{\ell}$ for some set $\setL$, for some collection of parity-check
  matrices $\{ \matrH_{\ell} \}_{\ell \in \setL}$ such that $\matrH =
  \sum_{\ell \in \setL} \matrH_{\ell}$ (in $\Z$), and for some collection of
  $M \times M$ permutation matrices $\{ \matrP_{\ell} \}_{\ell \in
    \setL}$. Moreover, let $\setI$ be the set of column indices of $\matrH$
  and let $\setI \times \setM$ with $\setM \defeq \{ 0, 1, \ldots, M-1 \}$ be
  the set of column indices of $\tmatrH$. With this, if $\tvomega = (\tilde
  \omega_{(i,m)})_{(i,m) \in \setI \times \setM}$ is a pseudo-codeword of
  $\tmatrH$, then $\vomega = (\omega_i)_{i \in \setI}$ with
  \begin{align}
    \omega_i
      &\defeq
         \frac{1}{M}
           \sum_{m \in \setM}
             \tilde \omega_{(i,m)}
             \quad \text{(in $\R$)}
               \label{eq:pcw:projection:1}
  \end{align}
  is a pseudo-codeword of $\matrH$.
\end{lemma}

\begin{proof} (Sketch.)  There are different ways to verify this
  statement. One approach is to show that, based on the fact that $\tvomega$
  satisfies the inequalities that define the fundamental polytope of
  $\tmatrH$~\cite{Koetter:Vontobel:03:ISTC, Vontobel:Koetter:05:IT:subm,
    Feldman:03:diss, Feldman:Wainwright:Karger:05:IT}, $\vomega$ satisfies the
  inequalities that define the fundamental polytope of $\matrH$. (We omit the
  details.)  Another approach is to use the fact that pseudo-codewords with
  rational entries are given by suitable projections of codewords in graph
  covers~\cite{Koetter:Vontobel:03:ISTC, Vontobel:Koetter:05:IT:subm}. So, for
  every pseudo-codeword $\tvomega$ of $\tmatrH$ with rational entries, there
  is some graph cover of $\graphT(\tmatrH)$ with a codeword in it, which, when
  projected down to $\graphT(\tmatrH)$, gives $\tvomega$. However, that graph
  cover of $\graphT(\tmatrH)$ is also a graph cover of $\graphT(\matrH)$, and
  so this codeword, when projected down to $\graphT(\matrH)$, gives $\vomega$
  as defined in~\eqref{eq:pcw:projection:1}. (We omit the details;
  see~\cite{Smarandache:Pusane:Vontobel:Costello:09:IT} for a similar, but
  less general, result.)
\end{proof}

One can then proceed as in~\cite{Smarandache:Pusane:Vontobel:Costello:09:IT}
and show that the AWGNC, the BSC, and the BEC
pseudo-weights~\cite{Wiberg:96:diss, Koetter:Vontobel:03:ISTC,
  Vontobel:Koetter:05:IT:subm, Feldman:03:diss,
  Feldman:Wainwright:Karger:05:IT,
  Forney:Koetter:Kschischang:Reznik:01:chapter} of $\tvomega$ will be at least
as large as the corresponding pseudo-weights of $\vomega$. As a corollary, the
minimum AWGNC, BSC, and BEC pseudo-weights of $\tmatrH$ are, respectively, at
least as large as the corresponding minimum pseudo-weights of
$\matrH$. Similar results can also be obtained for the minimum Hamming
distance.

Because the high-SNR behavior of linear programming decoding is dominated by
the minimum pseudo-weight of the relevant parity-check matrix, the high-SNR
behavior of linear programming decoding of the code defined by $\tmatrH$ is at
least as good as the high-SNR behavior of linear programming decoding of the
code defined by $\matrH$.\footnote{We neglect here the influence of the
  multiplicity of the minimum pseudo-weight pseudo-codewords.}

In general, because of the observations made in
Section~\ref{sec:cycleanalysis} about the ``breaking'' of cycles and the fact
that the active part of a pseudo-codeword must contain at least one cycle, it
follows that the unwrapping process is beneficial for the pseudo-codeword
properties of an unwrapped code, in the sense that many pseudo-codewords that
exist in the base code do not map to pseudo-codewords in the unwrapped
code. It is an intriguing challenge to better understand this process and its
influence on the low-to-medium SNR behavior of linear programming and
message-passing iterative decoders, in particular, to arrive at a better
analytical explanation of the significant gains that are visible in the
simulation plots that were shown in
Section~\ref{sec:variations:unwrapping:1}. To this end, the results
of~\cite{Lentmaier:Fettweis:Zigangirov:Costello:09:ITA}
and~\cite{Kudekar:Richardson:Urbanke:10:online} with respect to some related
code families (see the discussion in
Section~\ref{sec:other:LDPC:code:constructions:1}) will be very helpful, since
they indicate that some of the features of the fundamental polytope deserve
further analysis.

\subsection{Cost of the ``Convolutional Gain''}
\label{sec:cost}

In this subsection, we investigate the cost of the convolutional gain by
comparing several aspects of decoders for LDPC block and convolutional
codes. In particular, we consider the computational complexity, hardware
complexity, decoder memory requirements, and decoding delay. More details on
the various comparisons described in this section can be found in
\cite{Pusane:Jimenez:Sridharan:Lentmaier:Zigangirov:Costello:08:TCOM,
  Costello:Pusane:Bates:Zigangirov:06:ITA,
  Costello:Pusane:Jones:Divsalar:07:ITA}.

LDPC block code decoders and LDPC convolutional code decoders have the same
computational complexity per decoded bit and per iteration since LDPC
convolutional codes derived from LDPC block codes have the same node degrees
(row and column weights) in their Tanner graph representations, which
determines the number of computations required for message-passing decoding.

We adopt the notion of \emph{processor size} to characterize the hardware
complexity of implementing the decoder. A decoder's processor size is
proportional to the maximum number of variable nodes that can participate in a
common check equation. This is the block length $n$ for a block code, since
any two variable nodes in a block can participate in the same check
equation. For a convolutional code, this is the constraint length $\nus$,
since no two variable nodes that are more than $\nus$ positions apart can
participate in the same check equation. The constraint lengths of the LDPC
convolutional codes derived from LDPC block codes of length $n$ satisfy $\nus
\leq n$. Therefore, the convolutional codes have a processor size less than or
equal to that of the underlying block code.

On the other hand, the fully parallel pipeline decoding architecture penalizes
LDPC convolutional codes in terms of decoder memory requirements (and decoding
delay/latency) as a result of the $I$ iterations being multiplexed in space
rather than in time. The pipeline decoder architecture of Figure
\ref{fig:decoder} consists of $I$ identical processors of size $\nus$
performing $I$ decoding iterations simultaneously on independent sections of a
decoding window containing $I$ constraint lengths of received symbols. This
requires $I$ times more decoder memory elements than an LDPC block code
decoder that employs a single processor of size $n=\nus$ performing $I$
decoding iterations successively on the same block of received
symbols. Therefore, the decoder memory requirements and the decoding delay of
the pipeline decoder are proportional to $\nus \cdot I$, whereas the block
decoder's memory and delay requirements are only proportional to $n$. Another
way of comparing the two types of codes, preferred by some researchers, is to
equate the block length of a block code to the memory/delay requirements,
rather than the processor size, of a convolutional code, i.e., to set
$n=\nu_{\textrm{s}}\cdot I$. In this case the block code, now having a block
length many times larger than the constraint length of the convolutional code,
will typically (depending on $I$) outperform the convolutional code, but at a
cost of a much larger hardware processor. Finally, as noted in Section
\ref{sec:ldpc:convolutional:codes:1}, the parallel pipeline decoding
architecture for LDPC convolutional codes can be replaced by a serial looping
decoding architecture, resulting in fewer processors but a reduced throughput
along with the same memory and delay requirements.

In summary, the convolutional gain achieved by LDPC convolutional codes
derived from LDPC block codes comes at the expense of increased decoder memory
requirements and decoding delays. Although this does not cause problems for
some applications that are not delay-sensitive (e.g., deep-space
communication), for other applications that are delay-sensitive (e.g.,
real-time voice/video transmission), design specifications may be met by
deriving LDPC convolutional codes from shorter LDPC block codes, thus
sacrificing some coding gain, but reducing memory and delay requirements, or
by employing a reduced window size decoder, as suggested in the recent paper
by Papaleo \textit{et al.}~\cite{Papaleo:Iyengar:Siegel:Wolf:Corazza:10:ITW},
with a resulting reduction in the ``convolutional gain.''

\section{Conclusions}
\label{sec:discussion}

In this paper we showed that it is possible to connect two known techniques
for deriving LDPC convolutional codes from LDPC block codes, namely the
techniques due to Tanner and due to Jim\'enez-Feltstr\" om and Zigangirov.  This
connection was explained with the help of graph covers, which were also used
as a tool to present a general approach for constructing interesting classes
of LDPC convolutional codes. Because it is important to understand how the
presented code construction methods can be used --- and in particular combined
--- we then discussed a variety of LDPC convolutional code constructions,
along with their simulated performance results.

In the future, it will be worthwhile to extend the presented analytical
results, in particular to obtain a better quantitative understanding of the
low-to-medium SNR behavior of LDPC convolutional codes. In that respect, the
insights in the papers by Lentmaier \textit{et
  al.}~\cite{Lentmaier:Fettweis:Zigangirov:Costello:09:ITA,
  Lentmaier:Mitchell:Fettweis:Costello:10:ITA} and Kudekar \textit{et
  al.}~\cite{Kudekar:Richardson:Urbanke:10:online} on the behavior of related
code families will be valuable guidelines for further investigation.

\section*{Acknowledgments}

The authors would like to thank Chris Jones, Michael Lentmaier, David
Mitchell, Michael Tanner, and Kamil Zigangirov for their valuable discussions
and comments. We also gratefully acknowledge the constructive comments made by
the reviewers.

\end{document}